\pdfoutput=1
\RequirePackage{ifpdf}
\ifpdf 
\documentclass[pdftex]{sigma}
\else
\documentclass{sigma}
\fi

\numberwithin{equation}{section}

\newtheorem{Theorem}{Theorem}[section]
\newtheorem*{Theorem*}{Theorem}
\newtheorem{Corollary}[Theorem]{Corollary}
\newtheorem{Lemma}[Theorem]{Lemma}
\newtheorem{Proposition}[Theorem]{Proposition}
 { \theoremstyle{definition}
\newtheorem{Definition}[Theorem]{Definition}

\newtheorem{Remark}[Theorem]{Remark} }

\begin{document}
\allowdisplaybreaks

\newcommand{\arXivNumber}{2306.07062}

\renewcommand{\PaperNumber}{020}

\FirstPageHeading

\ShortArticleName{Non-Integrability of the Sasano System of Type $D_5^{(1)}$ and Stokes Phenomena}

\ArticleName{Non-Integrability of the Sasano System of Type $\boldsymbol{D_5^{(1)}}$\\ and Stokes Phenomena}

\Author{Tsvetana STOYANOVA}

\AuthorNameForHeading{Ts.~Stoyanova}

\Address{Faculty of Mathematics and Informatics, Sofia University ``St. Kliment Ohridski'', \\
5 J. Bourchier Blvd., Sofia 1164, Bulgaria}
\Email{\href{mailto:cveti@fmi.uni-sofia.bg}{cveti@fmi.uni-sofia.bg}}

\ArticleDates{Received November 28, 2023, in final form March 10, 2025; Published online March 27, 2025}

\Abstract{In 2006, Y.~Sasano proposed higher-order Painlev\'{e} systems, which admit affine Weyl group symmetry of type \smash{$D^{(1)}_l$}, $l=4, 5, 6, \dots$. In this paper, we study the integrability of a four-dimensional Painlev\'{e} system, which has symmetry under the extended affine Weyl group \smash{$\widetilde{W}\bigl(D^{(1)}_5\bigr)$} and which we call the Sasano system of type \smash{$D^{(1)}_5$}. We prove that one family of the Sasano system of type \smash{$D^{(1)}_5$} is not integrable by rational first integrals. We describe Stokes phenomena relative to a subsystem of the second normal variational equations. This approach allows us to compute in an explicit way the corresponding differential Galois group and therefore to determine whether the connected component of its unit element is not Abelian. Applying the Morales--Ramis--Sim\'{o} theory, we establish a non-integrable result.}


\Classification{34M55; 37J30; 34M40; 37J65}

\section{Introduction}

 In 2006, using the methods from algebraic geometry Y.~Sasano \cite{Sa} introduced higher-order Painlev\'{e}
 systems, which have symmetry under the affine Weyl group of type \smash{$D^{(1)}_l$}, $l=4, 5, 6, \dots$. Subsequently Fuji and Suzuki \cite{FS} derived the higher-order
 Painlev\'{e} systems of type \smash{$D^{(1)}_{2 n +2}$} from the Drinfeld--Sokolov hierarchy by similarity reduction.
 These higher-order Painlev\'{e} systems have four essential properties:
 \begin{enumerate}\itemsep=0pt
 \item They are Hamiltonian systems.
 	\item
 	They admit an affine Weyl group symmetry of type \smash{$D^{(1)}_l$} as B\"{a}cklund transformations.
 	\item
 	They can be considered as higher-order analogues of the Painlev\'{e} V and Painlev\'e VI systems.
 	\item
 	They have several symplectic coordinate systems, on which the Hamiltonians are polynomial.
 \end{enumerate}

 In this paper, we study the integrability of the following fourth-order Painlev\'{e} system
 \begin{gather}
 \dot{x} =
 \frac{2 x^2 y}{t} + x^2 - \frac{2 x y}{t}
 -\biggl(1 + \frac{\beta}{t}\biggr) x
 +\frac{\alpha_2+\alpha_5}{t} + \frac{2 z ((z-1) w + \alpha_3)}{t},\nonumber\\
 \dot{y} =
 -\frac{2 x y^2}{t} + \frac{y^2}{t} - 2 x y
 +\biggl(1 + \frac{\beta}{t}\biggr) y -\alpha_1,\nonumber\\
 \dot{z} =
 \frac{2 z^2 w}{t} + z^2 - \frac{2 z w}{t}
 -\biggl(1 + \frac{\alpha_5+\alpha_4}{t}\biggr) z
 +\frac{\alpha_5}{t} + \frac{2 y z (z-1)}{t},\nonumber\\
 \dot{w} =
 -\frac{2 z w^2}{t} + \frac{w^2}{t} - 2 z w
 +\biggl(1 + \frac{\alpha_5+\alpha_4}{t}\biggr) w - \alpha_3 -
 \frac{2 y (-w + 2 z w + \alpha_3)}{t}\label{s}
 \end{gather}
 with $\beta=2 \alpha_2+2 \alpha_3 + \alpha_4+\alpha_5$, where $\alpha_0, \alpha_1, \dots, \alpha_5$
 are complex parameters, which satisfy the relation
 \begin{equation*}
 \alpha_0+\alpha_1 + 2 \alpha_2 + 2 \alpha_3 + \alpha_4 + \alpha_5=1 .
 \end{equation*}
 The system~\eqref{s} is one of the systems introduced by Sasano in \cite{Sa}. It admits the extended affine Weyl
 group \smash{$\widetilde{W}\bigl(D^{(1)}_5\bigr)$} as a group of B\"{a}cklund transformations. For these reasons, throughout this
 paper we call the system~\eqref{s} the Sasano system of type \smash{$D^{(1)}_5$} or in short the Sasano system.
 The system~\eqref{s} is a two-degree-of-freedom non-autonomous Hamiltonian system with the Hamiltonian
 \begin{gather}
 H = H_{V} (x, y, t; \alpha_2+\alpha_5, \alpha_1, \alpha_2 + 2 \alpha_3 +\alpha_4) +
 H_{V} (z, w, t; \alpha_5, \alpha_3, \alpha_4) \nonumber\\
 \hphantom{H =}{} +
 \frac{2 y z ((z-1) w + \alpha_3)}{t} ,\label{H51}
 \end{gather}
 where $H_{V}(q, p, t; \gamma_1, \gamma_2, \gamma_3)$ is the Hamiltonian associated with the fifth Painlev\'{e} equation,
 i.e.,
 \[
 H_{V}(q, p, t ; \gamma_1, \gamma_2, \gamma_3)=
 \frac{q (q-1) p (p+t) - (\gamma_1+\gamma_3) q p + \gamma_1 p + \gamma_2 t q}{t} .
 \]
 Hence the system~\eqref{s} is considered as coupled Painlev\'{e} V systems in dimension 4. The
 non-autonomous Hamiltonian system~\eqref{s} can be turned into an autonomous one with three degrees of freedom
 by introducing two new dynamical variables: $t$ and its conjugate variable $-F$. The~new Hamiltonian becomes
 \[
 \widetilde{H}=H+F ,
 \]
 where $H$ is given by~\eqref{H51}. Then the extended Hamiltonian system~\eqref{s} becomes
 \begin{alignat}{3}
 &\frac{{\rm d} x}{{\rm d} s}=\frac{\partial \widetilde{H}}{\partial y},\qquad &&
 \frac{{\rm d} y}{{\rm d} s}=-\frac{\partial \widetilde{H}}{\partial x},\nonumber&\\
 &\frac{{\rm d} z}{{\rm d} s}=\frac{\partial \widetilde{H}}{\partial w},\qquad &&
 \frac{{\rm d} w}{{\rm d} s}=-\frac{\partial \widetilde{H}}{\partial z},\nonumber&\\
 &\frac{{\rm d} t}{{\rm d} s}=\frac{\partial \widetilde{H}}{\partial F},\qquad &&
 \frac{{\rm d} F}{{\rm d} s}=-\frac{\partial \widetilde{H}}{\partial t} .&\label{es}
 \end{alignat}
 The symplectic structure $\omega$ is canonical in the variables $(x, y, z, w, t, F)$, i.e.,
 $\omega={\rm d} x \wedge {\rm d} y + {\rm d} z \wedge {\rm d} w + {\rm d} t \wedge {\rm d} F$.

 In this paper, we are interested in the non-integrability
 of the Hamiltonian system~\eqref{es}.
 Recall that from the theorem of Liouville--Arnold \cite{Ar} this means the non-existence of
 three first integrals $f_1=\widetilde{H}, f_2, f_3$ functionally independent and in involution.
 We prove that the Sasano system~\eqref{s} is non-integrable by rational first integrals. Our approach comes under the frame of the Morales--Ramis--Sim\'{o} theory. This theory
 reduces the problem of integrability of a given analytic Hamiltonian system to the problem of integrability of
 the variational equation of this Hamiltonian system along one particular non-stationary solution. Since the
 variational equations are linear ordinary differential equations their integrability is well defined in the context
 of the differential Galois theory. The Morales--Ramis--Sim\'o theory finds applications in the study of non-integrability of a huge range of dynamical systems like
 $N$-body problems~\mbox{\cite{BC, Co3, MPSS, PM, Ts1, Ts2}}, problems with homogeneous potentials \cite{CDMP, Co1, CMP}, the Painlev\'e equation and their $q$-analogues~\mbox{\cite{CR, F, M, St1, St2, St3}},
 the higher-order analogues of Painlev\'e systems \cite{St4},
 as well as in the study of non-integrability of non-Hamiltonian systems \cite{AZ, C}, in the study of integrability in the Jacobi sense \cite{MP}, in the study of the irreducibility
of the Painlev\'e equations \cite{CW}, etc.

 In this paper, we study the extended Sasano system~\eqref{es} when $\alpha_1=\alpha_2=\alpha_3=0$, ${\alpha_4=1}$, ${\alpha_0=-\alpha_5}$.
 It turns out that the differential Galois group of the first
 variational equations is a~commutative group.
 To establish a non-integrable result, we find an obstruction to integrability studying the linearized second normal variational equations
 $({\rm LNVE})_2$, which is a linear homogeneous system of thirteen order. To determine the Galois group of such a higher-order linear system, we find
 a subsystem of the $({\rm LNVE})_2$, whose Galois group can be computed explicitly.
 It turns out that this subsystem is a system with non-trivial Stokes phenomena at the infinity point. Computing the corresponding
 Stokes matrices we deduce that the connected component of the unit element of the differential Galois group
 of this subsystem and hence the Galois group of the $({\rm LNVE})_2$ is not an Abelian group. Then
 the key result of this paper (Theorem \ref{main-s} in Section~\ref{section3}) states

 \begin{Theorem}\label{key}
 Assume that $\alpha_1=\alpha_2=\alpha_3=0$, $\alpha_4=1$, $\alpha_0=-\alpha_5$, where $\alpha_5$ is arbitrary. Then the
 Sasano system~\eqref{s} is not integrable in the Liouville--Arnold sense by rational first integrals.
 \end{Theorem}

 Using B\"{a}cklund transformations of the Sasano system~\eqref{s}, which are rational canonical transformations~\cite{SY}, we can extend the result of the key Theorem~\ref{key} to
 the main results of this paper (Theorems \ref{1} and~\ref{2} in Section~\ref{section4}).
 \begin{Theorem}\label{main}
 Let $\alpha$ be an arbitrary complex parameter, which is not an integer. Assume that the parameters $\alpha_j$ are either of the kind
 $\pm \alpha + n_j$ or of the kind $l_j, n_j, l_j\in\mathbb{Z}$ in such a way that~${1-\alpha_0-\alpha_1}$ and
 $\alpha_4+\alpha_5$ are together of the kind $\pm \alpha + m_i$, $m_i\in\mathbb{Z}$, $i=1, 2$. Then the Sasano system
 \eqref{s} is not integrable in the Liouville--Arnold sense by rational first integrals.
 \end{Theorem}

 \begin{Theorem}\label{main0}
 Assume that all of the parameters $\alpha_j$, $0 \leq j \leq 5$, are integer in such a way that~$1-\alpha_0-\alpha_1$
 and $\alpha_4+\alpha_5$ are together either even or odd integer. Then the Sasano system~\eqref{s}
 is not integrable in the Liouville--Arnold sense by rational first integrals.
 \end{Theorem}
 In fact, the result of Theorem~\ref{main} contains the result of Theorem~\ref{main0}. We present Theorem~\ref{main0} as an independent
 result because of the additional specification of the quantities $1-\alpha_0-\alpha_1$ and $\alpha_4+\alpha_5$ when $\alpha\in\mathbb{Z}$.

 This paper is organized as follows. In the next section, we briefly review the basics of the Morales--Rams--Sim\'{o} theory
 of the non-integrability of the Hamiltonian systems and the relation of the differential Galois theory to the linear
 systems of ordinary differential equations. In~Section~\ref{section3}, we prove non-integrability of the Sasano system~\eqref{s} when
 $\alpha_1=\alpha_2=\alpha_3=0$, $\alpha_4=1$ and $\alpha_0=-\alpha_5$. In~Section~\ref{section4}, using B\"{a}cklund transformations
 of the Sasano system~\eqref{s}, we extend the result of the Section~\ref{section3} to the entire orbits of the parameters $\alpha_j$
 and establish the main theorems of this paper.

 \section{Preliminaries}

 \subsection{Non-integrability of Hamiltonian systems and differential Galois theory}
 In this subsection, we briefly recall Morales-Ruiz--Ramis--Sim\'{o} theory of non-integrability of Hamiltonian systems following
 \cite{MR,MR2, MR1, MRS}.

Let $M$ be a symplectic analytical complex manifold of complex dimension $2 n$. Consider on~$M$ a Hamiltonian system
 \begin{equation}
 \dot{x}=X_H(x)\label{H}
 \end{equation}
 with a Hamiltonian $H \colon M \rightarrow \mathbb{C}$. Let $x(t)$ be a particular solution of~\eqref{H}, which is not an equilibrium point.
 Denote by $\Gamma$ the phase curve corresponding to this solution. The first variational equations $({\rm VE})_1$ of
 \eqref{H} along $\Gamma$ are written
 \begin{equation}
 \dot{\xi}=\frac{\partial X_H}{\partial x} (x(t)) \xi,\qquad \xi\in T_{\Gamma}M .\label{fve0}
 \end{equation}
 Using the Hamiltonian $H$, we can always reduce the degrees of freedom of the variational equations~\eqref{fve0} by one in the
 following sense. Consider the normal bundle of $\Gamma$ on the level variety $M_h=\{ x \mid H(x)=h\}$. The projection of the
 variational equations~\eqref{fve0} on this bundle induces the so called first normal variational equations $({\rm NVE})_1$ along $\Gamma$.
 The dimension of the $({\rm NVE})_1$ is~${2 n-2}$. Assume now that $x(t)$ is a rational non stationary particular solution of~\eqref{H} and let as above
 $\Gamma$ be the phase curve corresponding to it. Assume also that the field $K$ of the coefficients of the $({\rm NVE})_1$ is the field of rational functions in $t$,
 that is $K=\mathbb{C}(t)$. Assume also that $t=\infty$ is an irregular singularity for the $({\rm NVE})_1$.
 The entries of a fundamental matrix solution of $({\rm NVE})_1$ define a Picard--Vessiot extension $L_1$ of the field $K$.
 This in its turn defines a~differential Galois group
 $G_1={\rm Gal}(L_1/K)$.
 Then the main theorem of the Morales-Ruiz--Ramis theory states \cite{MR, MR2, MRS}.

\begin{Theorem}[Morales--Ramis]
 Assume that the Hamiltonian system~\eqref{H} is completely integrable with rational first integrals in a
 neighbourhood of $\Gamma$, not necessarily independent on $\Gamma$ itself. Then the identity component $(G_1)^0$ of the
 differential Galois group $G_1={\rm Gal}(L_1/K)$ is Abelian.
 \end{Theorem}

 The problem considered in this paper is one among many examples illustrating that the opposite is not true in general.
 That is if the connected component $(G_1)^0$ of the unit element of the differential Galois group ${\rm Gal}(L_1/K)$ is
 Abelian, one cannot deduce that the corresponding Hamiltonian system~\eqref{H} is completely integrable. Beyond the first
 variational equations Morales-Ruiz, Ramis and Sim\'{o} suggest in \cite{MRS} to use higher-order variational equations
 to solve such integrability problems. Let as above $x(t)$ be a particular rational non stationary solution of the Hamiltonian
 system~\eqref{H}. We write the general solution as $x(t, z)$, where $z$ parametrizes it near $x(t)$ as $x(t, z_0)=x(t)$.
 Then we can write the system~\eqref{H} as
 \begin{equation}
 \dot{x}(t, z)=X_H(x(t, z)) .\label{H1}
 \end{equation}
 Denote by \smash{$x^{(k)}(t, z)$}, $k \geq 1$ the derivatives of $x(t, z)$ with respect to $z$ and by
 \smash{$X^{(k)}_H(x)$}, $k \geq 1$ the derivatives of $X_H(x)$ with respect to~$x$. By successive derivations of~\eqref{H1} with respect
 to~$z$ and evaluations at $z_0$, we obtain the so called $k$-th variational equations $({\rm VE})_k$ along the solution~$x(t)$%
 \begin{equation}
 \dot{x}^{(k)}(t)=X^{(1)}_H(x(t)) x^{(k)}(t) + P\bigl(x^{(1)}(t), x^{(2)}(t), \dots, x^{(k-1)}(t)\bigr) .\label{vek}
 \end{equation}
 Here $P$ denotes polynomial terms in the monomials of order $|k|$ of the components of its arguments. The coefficients
 of $P$ depend on $t$ through \smash{$X^{(j)}_H(x(t))$}, $j < k$. For every $k > 1$, the linear non-homogeneous system~\eqref{vek}
 can be arranged as a linear homogeneous system of higher dimension by making the monomials of order $|k|$ in $P$ new
 variables and adding to~\eqref{vek} their differential equations. If we restrict the system~\eqref{vek} to the variables
 that define the $({\rm NVE})_1$, the corresponding linear homogeneous system is the so called $k$-th linearized
 normal variational equations $({\rm LNVE})_k$. The solutions of the chain of $({\rm LNVE})_k$ define a chain
 of Picard--Vessiot extensions of the main field $K=\mathbb{C}(t)$ of the coefficients of $({\rm NVE})_1$, i.e., we have
 $K \subset L_1 \subset L_2 \subset \cdots \subset L_k$, where $L_1$ is above, $L_2$ is the Picard--Vessiot extension
 of $K$ associated with $({\rm LNVE})_2$, etc. Then we can define the differential Galois groups
 $G_1={\rm Gal}(L_1/K), G_2={\rm Gal}(L_2/K), \dots, G_k={\rm Gal}(L_k/K)$. Assume as above that $t=\infty$ is an irregular singularity for the $({\rm NVE})_1$
 and therefore for $({\rm NVE})_k$ for all $k \geq 2$. Then the main theorem of the
 Morales-Ruiz--Ramis--Sim\'{o} theory states the following.

 \begin{Theorem}[Morales--Ramis--Sim\'{o}]\label{MRS}
 If the Hamiltonian system~\eqref{H} is completely integrable with rational first integrals, then for every $k\in\mathbb{N}$
 the connected component of the unit element~$(G_k)^0$ of the differential Galois group $G_k={\rm Gal}(L_k/K)$
 is Abelian.
 \end{Theorem}

 From Theorem \ref{MRS}, it follows that if we find a group $(G_k)^0$, which is not Abelian, then the Hamiltonian system~\eqref{H} will be non-integrable
 by means of rational first integrals.
 Note that this non-commutative group $(G_k)^0$ will be a solvable group. In this way,
 non-integrability in the sense of the Hamiltonian dynamics will correspond to integrability in the Picard--Vessiot sense.

 \subsection[Differential Galois group of a linear system of ordinary differential equations]{Differential Galois group of a linear system \\ of ordinary differential equations}

 In this subsection, we briefly recall some facts, notations and definitions from the differential Galois theory, needed
 to compute the differential Galois group of a linear system with one irregular and one regular singularity. We follow
 the works of van der Put, Mitschi, Singer and Ramis \cite{Mi,R3, Si, vPS}. Throughout this paper, all angular directions
 and sectors are defined on the Riemann surface of the natural logarithm.

 Consider a linear system of ordinary differential equations of order $n$
 \begin{equation}
 \dot{\upsilon}=A(t) \upsilon ,\label{ls1}
 \end{equation}
 where $A(t)\in {\rm GL}_n(\mathbb{C}(t))$.

 \begin{Definition}\label{Gal}
The differential Galois group $G$ of the system~\eqref{ls1} over $\mathbb{C}(t)$ is the group of all differential $\mathbb{C}(t)$-automorphisms of
 a Picard--Vessiot extension of $\mathbb{C}(t)$ relative to~\eqref{ls1}. This group is isomorphic to an algebraic subgroup of ${\rm GL}_n(\mathbb{C})$
 with respect to a fundamental matrix solution of~\eqref{ls1}.
 \end{Definition}

 Assume that the system~\eqref{ls1} has two singular points over $\mathbb{C}\mathbb{P}^1$ taken at $t=0$ and $t=\infty$. Assume that
 the origin is a regular singularity while $t=\infty$ is a
 non-resonant irregular singularity of
 Poincar\'{e} rank 1. Denote by $S$ the set of singular points of the system~\eqref{ls1}, that is, $S=\{0, \infty\}$.
 If we replace in Definition \ref{Gal} the field $\mathbb{C}(t)$ with the field of germs of meromorphic functions at $a\in S$, we define the so called local differential Galois
 group $G_a$ of~\eqref{ls1}. In what follows, we present effective theorems for computing the local differential Galois groups $G_a$, $a\in S$ of the system
 \eqref{ls1}.

 Let $\Phi(t)$ be a local fundamental matrix solution near the origin of the system~\eqref{ls1}.
 The~following result of Schlesinger \cite{Sch} describes the local differential Galois group at the origin of the system~\eqref{ls1}.

 \begin{Theorem}[Schlesinger]\label{reg}
 Under the above assumptions the monodromy group around the origin with respect to the fundamental matrix solution
 $\Phi(t)$ is a Zariski dense subgroup of the differential Galois group~$G$ of the
 system~\eqref{ls1}.
 \end{Theorem}

Since we prefer to work with an irregular singularity at the origin to at $t=\infty$, we make the change $t=1/\tau$ in the
 system~\eqref{ls1}.
 This transformation takes the system~\eqref{ls1} into the system
 \begin{equation}
 \upsilon'=A(\tau) \upsilon,\qquad '=\frac{{\rm d}}{{\rm d} \tau} ,\label{ls}
 \end{equation}
 for which the origin is a non-resonant irregular singularity of Poincar\'{e} rank~1.
 Denote by $\mathbb{C}(\tau)$, $\mathbb{C}((\tau))$ and~$\mathbb{C}\{\tau\}$ the differential fields of rational functions, formal power series and convergent
 power series, respectively. Note that
 \[
 \mathbb{C}(\tau) \subset \mathbb{C}\{\tau\} \subset \mathbb{C}((\tau)) .
 \]
 In what follows, we determine the local differential Galois groups of the system~\eqref{ls} around the~origin.

 From the Hukuhara--Turrittin theorem \cite{W}, it follows that the
 system~\eqref{ls} admits a formal fundamental matrix solution at the origin of the form
 \begin{equation}
 \hat{\Psi}(\tau)=\hat{H}(\tau) \tau^{\Lambda} \exp\biggl(\frac{Q}{\tau}\biggr) ,\label{fls}
 \end{equation}
 where
 \[
 \Lambda={\rm diag} (\lambda_1, \lambda_2, \dots, \lambda_n),\qquad
 Q={\rm diag} (q_1, q_2, \dots, q_n),\qquad \hat{H}(t)\in{\rm GL}_n(\mathbb{C}((t)))
 \]
 with $\lambda_j$, $q_j \in\mathbb{C}$, $j=1, \dots, n$. Consider the system~\eqref{ls} and its formal fundamental matrix solution
 $\hat{\Psi}(\tau)$ over the field $\mathbb{C}((\tau))$.

 \begin{Definition}
 {With respect to the formal fundamental matrix solution $\hat{\Psi}(\tau)$ from~\eqref{fls}, we define the formal
 monodromy matrix $\hat{M}_0\in{\rm GL}_n(\mathbb{C})$ around the origin as
 \[
 \hat{\Psi}\bigl(\tau. {\rm e}^{2 \pi  {\rm i}}\bigr)=\hat{\Psi}(\tau) \hat{M}_0 .
 \]
 In particular,
 \[
 \hat{M}_0={\rm e}^{2 \pi {\rm i} \Lambda} .
 \]}
 \end{Definition}

 \begin{Definition}
 With respect to the formal fundamental matrix solution $\hat{\Psi}(\tau)$ from~\eqref{fls} we define the exponential
 torus $\mathcal{T}$ as the differential Galois group ${\rm Gal}(E/F)$, where
 \[
 F=\mathbb{C}((\tau))\bigl(\tau^{\lambda_1}, \tau^{\lambda_2}, \dots, \tau^{\lambda_n}\bigr) \qquad \text{and} \qquad
 E=F\bigl({\rm e}^{q_1/\tau}, {\rm e}^{q_2/\tau}, \dots, {\rm e}^{q_n/\tau}\bigr).
 \]
 \end{Definition}
 We may consider $\mathcal{T}$ as a subgroup of $(\mathbb{C}^*)^n$. The Zariski closure of the group
 generated by the formal monodromy matrix and exponential torus yields the so called formal differential Galois group
 at the origin of the system~\eqref{ls} (see \cite{Si, vPS}).

 Consider now the system~\eqref{ls} over the field $\mathbb{C}\{\tau\}$. In general, the entries $\hat{h}_{ij}(\tau)$, $1 \leq i, j \leq n$
 of the matrix \smash{$\hat{H}(\tau)$} in~\eqref{fls} are either divergent or convergent power series in $\tau$.
 The existence of divergent power series entries in \smash{$\hat{H}(\tau)$} ensures an observation of a non-trivial Stokes
 phenomenon at the origin.

 \begin{Definition}\label{st}
 Under the above notations for every divergent power series $\hat{h}_{i j}(\tau)$, we define a~set~$\Theta_j$
 	of admissible singular directions $\theta_{ji}$, $0 \leq \theta_{ji} < 2 \pi$, where $\theta_{ji}$ is the
 	bisector of the maximal angular sector \smash{$\bigl\{{\rm Re} \bigl(\frac{q_i-q_j}{\tau}\bigr) < 0 \bigr\}$}. In
 	particular,
 	 \[
 	 \Theta_j=\{ \theta_{ji},\, 0 \leq \theta_{ji} < 2 \pi,\, \theta_{ji}= \arg (q_j-q_i),\,
 	 1 \leq i, j  \leq n,\, i \neq j \} .
 	 \]
 \end{Definition}
 In order to compute the analytic invariants at the origin of the system~\eqref{ls}, we have to lift the formal
 fundamental matrix solution $\hat{\Psi}(\tau)$ from~\eqref{fls} to such an actual. To solve this problem,
 in this paper we utilize the summability theory. The application of the summability theory to ordinary
 differential equations generalizes the theorem of Hukuhara--Turrittin to the following theorem of Ramis~\cite{R2}.

 \begin{Theorem}\label{R}
 In the formal fundamental matrix solution at the origin $\hat{\Psi}(\tau)$ from~\eqref{fls} the entries of the
 matrix \smash{$\hat{H}(\tau)$} are $1$-summable along any non-singular direction $\theta$. If we denote by $H_{\theta}(\tau)$
 the $1$-sum of the matrix \smash{$\hat{H}(\tau)$} along $\theta$, then
 $\Psi_{\theta}(\tau)=H_{\theta}(\tau) \tau^{\Lambda} \exp(Q/\tau)$ is an actual fundamental matrix solution at the origin
 of the system~\eqref{ls}.
 \end{Theorem}
 Since this paper is not devoted to the summability theory, rather we only use it, we will not consider it in details.
 For the needed facts, notation and definitions, we refer to the works of Loday-Richaud \cite{LR}, as well as
 the works of Ramis \cite{R3, R2}.

 Let $\varepsilon > 0$ be a small number. Let $\theta-\varepsilon$ and $\theta+\varepsilon$ be two non-singular neighboring
 directions to the singular direction $\theta\in \Theta_j$. Let $\Psi_{\theta-\varepsilon}(\tau)$ and $\Psi_{\theta+\varepsilon}(\tau)$
 be the actual fundamental matrix solutions at the origin of the system~\eqref{ls} corresponding to
 the directions $\theta-\varepsilon$ and $\theta+\varepsilon$ in the sense of Theorem~\ref{R}.
 \begin{Definition}\label{Sto}
With respect to the actual fundamental matrix solutions $\Psi_{\theta-\varepsilon}(\tau)$ and $\Psi_{\theta+\varepsilon}(\tau)$
 	the Stokes matrix $St_{\theta}\in{\rm GL}_n(\mathbb{C})$ related to the singular direction $\theta$ is defined as
 	 \[
 	 St_{\theta}=\left(\Psi_{\theta+\varepsilon}(\tau)\right)^{-1} \Psi_{\theta-\varepsilon}(\tau) .
 	 \]
 \end{Definition}

 The next theorem of Ramis \cite{R3} determines the differential Galois group at the origin of the system~\eqref{ls} over
 the field $\mathbb{C}(\tau\}$.
 \begin{Theorem}[Ramis]\label{Ra}
 The differential Galois group at the origin of the system~\eqref{ls} over $\mathbb{C}\{\tau\}$ is the Zariski closure
 of the group generated by the formal differential Galois group at the origin and the collection of the Stokes
 matrices $\{St_{\theta}\}$ for all singular directions $\theta$.
 \end{Theorem}
 For more details about the relation between the Stokes phenomenon and the differential Galois theory, we refer to the very recent
 work of Ramis \cite{R1}.
 We make note that one can introduce the differential Galois group at $t=\infty$ of the system~\eqref{ls1} in the same way.

 Let $t_0$ be a base point of $\mathbb{C}\mathbb{P}^1 \backslash S$ and
 let $\Sigma_{t_0}$ denote an analytic germ of a fundamental matrix solution of~\eqref{ls1} at $t_0$. Let $U_a$, $a\in S$, be an open disc with center $a$,
 together with a local parameter $t_a$ at $a$, and such that $U_a \cap S=\{a\}$. Let $d_a$ be a fixed ray from $a$ in $U_a$, together with a point
$b_a\in d_a$ in $U_a$ and a path $\gamma_a$ from $t_0$ to $b_a$. Analytic continuation of $\Sigma_{t_0}$ along~$\gamma_a$ and~$d_a$ provides
an analytic germ $\Sigma_a$ of fundamental matrix solution on a germ of open sector with vertex $a$, bisected by $d_a$. Let $G_a$ be the local differential Galois group
of the system~\eqref{ls1} over the field of germs of meromorphic function at $a$ with respect to $\Sigma_a$. If we conjugate elements of $G_a$ by the analytic continuation
described above, we get an injective morphism of algebraic groups $G_a \hookrightarrow G$ with respect to the representation of these groups in
 ${\rm GL}_n(\mathbb{C})$ given by $\Sigma_a$ and $\Sigma_{t_0}$, respectively. In this way all $G_a$, $a\in S$, can be simultaneously identified with closed subgroups of $G$.
 Then we have the following important result of Mitschi \cite[Proposition 1.3]{Mi}.

 \begin{Theorem}[Mitschi]\label{Mi}
 The differential Galois group $G$ of the system~\eqref{ls1} is topologically generated in ${\rm GL}_n(\mathbb{C})$
 by the local differential Galois groups $G_a$, where $a$ runs over $S$.
 \end{Theorem}

 \section[Non-integrability for alpha\_1=alpha\_2=alpha\_3=0, alpha\_4=1, alpha\_0=-alpha\_5]{Non-integrability for $\boldsymbol{\alpha_1=\alpha_2=\alpha_3=0}$, $\boldsymbol{\alpha_4=1}$, $\boldsymbol{\alpha_0=-\alpha_5}$}\label{section3}

 In this section, we deal with the non-integrability of the Hamiltonian system~\eqref{es} when
 $\alpha_1=\alpha_2=\alpha_3=0$, $\alpha_4=1$, $\alpha_0=-\alpha_5$. Denote $\alpha=\alpha_5$. For these values of the parameters
 the autonomous Hamiltonian system~\eqref{es} becomes
 \begin{gather*}
 \frac{{\rm d} x}{{\rm d} s}
 =
 \frac{2 x^2 y}{t} + x^2 - \frac{2 x y}{t}
 -\biggl(1 + \frac{1+\alpha}{t}\biggr) x + \frac{\alpha}{t}
 + \frac{2 z (z-1) w }{t},\\
 \frac{{\rm d} y}{{\rm d} s} =
 -\frac{2 x y^2}{t} + \frac{y^2}{t} - 2 x y
 +\biggl(1 + \frac{1+\alpha}{t}\biggr) y ,\\
 \frac{{\rm d} z}{{\rm d} s} =
 \frac{2 z^2 w}{t} + z^2 - \frac{2 z w}{t}
 -\biggl(1 + \frac{1+\alpha}{t}\biggr) z + \frac{\alpha}{t}
 + \frac{2 y z (z-1)}{t},\\
 \frac{{\rm d} w}{{\rm d} s} =
 -\frac{2 z w^2}{t} + \frac{w^2}{t} - 2 z w
 +\biggl(1 + \frac{1+\alpha}{t}\biggr) w -
 \frac{2 y (-w + 2 z w )}{t}, \\
 \frac{{\rm d} t}{{\rm d} s} =
 1, \\
 \frac{{\rm d} F}{{\rm d} s} =
 \frac{1}{t^2} (x (x-1) y (y+t) + z (z-1) w (w+t) - (1+\alpha) (x y + z w)
 + \alpha (y+w)) \\
 \hphantom{\frac{{\rm d} F}{{\rm d} s} =}{} - \frac{x (x-1) y}{t} - \frac{z (z-1) w}{t} .
 \end{gather*}
 We choose
 \begin{equation}
 x=z=\frac{\alpha}{s}, \qquad y=w=0, \qquad t=s, \qquad F=0\label{ps}
 \end{equation}
 as a non-equilibrium particular solution, along which we will write the variational equations. Because of the equation
 \smash{$\frac{{\rm d} t}{{\rm d} s}=1$}, from here on we use $t$ instead of $s$.

 For the first normal variational equations $({\rm NVE})_1$ along the solution~\eqref{ps}, we obtain the system
 \begin{gather}
 \dot{x}_1 =
 \biggl(-1 + \frac{\alpha-1}{t}\biggr) x_1 + \biggl(\frac{2 \alpha^2}{t^3} - \frac{2 \alpha}{t^2}\biggr) y_1
 + \biggl(\frac{2 \alpha^2}{t^3} - \frac{2 \alpha}{t^2}\biggr) w_1,\nonumber\\
 \dot{z}_1 =
 \biggl(-1 + \frac{\alpha-1}{t} \biggr) z_1 +
 \biggl(\frac{2 \alpha^2}{t^3} - \frac{2 \alpha}{t^2}\biggr) w_1
 + \biggl(\frac{2 \alpha^2}{t^3} - \frac{2 \alpha}{t^2}\biggr) y_1,\nonumber\\
 \dot{w}_1 =
 \biggl(1 - \frac{\alpha-1}{t}\biggr) w_1,\nonumber\\
 \dot{y}_1 =
 \biggl(1 - \frac{ \alpha -1}{t} \biggr) y_1 .\label{fve}
 \end{gather}
 Note that the $({\rm NVE})_k$, $k\in\mathbb{N}$ of the system~\eqref{es} along the solution~\eqref{ps} are nothing but the $({\rm VE})_k$, $k\in\mathbb{N}$, of the
 system~\eqref{s} along the solution $x=z=\frac{\alpha}{t}$, $y=w=0$ for ${\alpha_1=\alpha_2=\alpha_3=0}$, $\alpha_4=1$, $\alpha_0=-\alpha_5=-\alpha$.

 The system~\eqref{fve} is solvable in quadratures and therefore its differential Galois group $G_1$ is a solvable subgroup
 in ${\rm GL}_4(\mathbb{C})$.

 \begin{Theorem}
 The connected component $(G_1)^0$ of the unit element of the differential Galois group $G_1$ of the $({\rm NVE})_1$ is Abelian.
 \end{Theorem}

 \begin{proof}

 The $({\rm NVE})_1$ have two singular points over $\mathbb{C}\mathbb{P}^1$: the points $t=0$ and $t=\infty$. The~origin is a regular singularity while $t=\infty$ is an irregular singularity.
 The system~\eqref{fve} admits a~fundamental matrix solution $\Phi(t)$ in the form
 \[
 \Phi(t)=\begin{pmatrix}
 {\rm e}^{-t} t^{\alpha-1} & 0 & -\alpha {\rm e}^t t^{-\alpha-1} & -\alpha {\rm e}^t t^{-\alpha-1}\\[0.1ex]
 0 & {\rm e}^{-t} t^{\alpha-1} & -\alpha {\rm e}^t t^{-\alpha-1} & -\alpha {\rm e}^t t^{-\alpha-1}\\[0.1ex]
 0 & 0 & {\rm e}^t t^{-\alpha+1} & 0\\[0.1ex]
 0 & 0 & 0 & {\rm e}^t t^{-\alpha+1}
 \end{pmatrix}.
 \]
 We will compute the Galois group $G_1$ of the $({\rm NVE})_1$ with respect to this fundamental matrix solution.
 In this case, from Theorem \ref{Mi}, it follows that
 the differential Galois group $G_1$ of the system~\eqref{fve} is generated topologically by the local Galois groups $G_0$ and $G_{\infty}$, corresponding to the singularities $t=0$ and $t=\infty$, respectively.

 Let us first determine the group $G_0$.
 In a neighborhood of the origin, the above solution $\Phi(t)$ is written as
 \[
 \Phi(t)=P(t) t^A ,
 \]
 where $P(t)$ is the holomorphic matrix
 \[
 P(t)=\begin{pmatrix}
 {\rm e}^{-t} & 0 & -\alpha {\rm e}^t & -\alpha {\rm e}^t\\
 0 & {\rm e}^{-t} & - \alpha {\rm e}^t & -\alpha {\rm e}^t \\
 0 & 0 & t^2 {\rm e}^t & 0\\
 0 & 0 & 0 & t^2 {\rm e}^t
 \end{pmatrix} .
 \]
 For the constant matrix $A$ we have that $A={\rm diag} (\alpha-1, \alpha-1, -\alpha-1, -\alpha-1)$.
 From Theorem~\ref{reg}, it follows that the Galois group $G_0$ over $\mathbb{C}(t)$ is generated topologically by the monodromy matrix~$M_0$ around the origin.
 With respect to the fundamental matrix solution $\Phi(t)$, we obtain
 \[
 M_0={\rm e}^{2 \pi {\rm i} A}=\begin{pmatrix}
 {\rm e}^{2 \pi {\rm i} \alpha} & 0 & 0 & 0\\
 0 & {\rm e}^{2 \pi {\rm i} \alpha} & 0\\
 0 & 0 & {\rm e}^{-2 \pi {\rm i} \alpha} & 0\\
 0 & 0 & 0 & {\rm e}^{-2 \pi {\rm i} \alpha}
 \end{pmatrix} .
 \]
 When $\alpha\in\mathbb{Q}$ but $\alpha\notin\mathbb{Z}$, the group generated by $M_0$ is not connected but it is a finite and cyclic group. In this case,
 \[
 G_0=\left\{\begin{pmatrix}
 \nu^{-1} & 0 & 0 & 0\\
 0 & \nu^{-1} & 0 & 0\\
 0 & 0 & \nu & 0\\
 0 & 0 & 0 & \nu
 \end{pmatrix}, \ \nu \ \text{is a root of unity} \right\}, \qquad (G_0)^0=\{I_4\} ,
 \]
 where $I_4$ is the identity matrix. When $\alpha\in\mathbb{Z}$, we have that $G_0=(G_0)^0=\{I_4\}$. When $\alpha\notin\mathbb{Q}$, the group generated by $M_0$ is a connected group
 and
 \[
 G_0=(G_0)^0=\left\{ \begin{pmatrix}
 \nu^{-1} & 0 & 0 & 0\\
 0 & \nu^{-1} & 0 & 0\\
 0 & 0 & \nu & 0\\
 0 & 0 & 0 & \nu
 \end{pmatrix},\ \nu\in\mathbb{C} \right\} .
 \]

 Consider now the $({\rm NVE})_1$ and the fundamental matrix solution $\Phi(t)$ at $t=\infty$. In a neighborhood of the irregular singularity $t=\infty$, the matrix $\Phi(t)$ is written
 as
 \[
 \Phi(t)=H_1(t) t^{\Lambda_1} \exp (Q_1 t) ,
 \]
 where $H_1(t)$ is the holomorphic matrix
 \[
 H_1(t)=\begin{pmatrix}
 1 & 0 & -\alpha t^{-2} & -\alpha t^{-2}\\
 0 & 1 & -\alpha t^{-2} & -\alpha t^{-2}\\
 0 & 0 & 1 & 0\\
 0 & 0 & 0 & 1
 \end{pmatrix} .
 \]
 The matrices $\Lambda_1$ and $Q_1$ are given by
 \[
 \Lambda_1={\rm diag} (\alpha-1, \alpha-1, -\alpha, -\alpha),\qquad
 Q_1={\rm diag}(-1, -1, 1, 1) .
 \]
 Since we do not observe non-trivial Stokes phenomena, the local Galois group $G_{\infty}$ is generated topologically by the formal monodromy $\hat{M}_{\infty}$ and the
 exponential torus $\mathcal{T}_{\infty}$. The formal monodromy corresponding to a loop around $t=\infty$ is nothing but $(M_0)^{-1}$, that is, $\hat{M}_{\infty}=(M_0)^{-1}$. Therefore, we can consider
 the local differential Galois group $G_0$ at the origin as a subgroup of the local differential Galois group $G_{\infty}$ at $t=\infty$. Thus the Galois group $G_1$ of the $({\rm NVE})_1$
 coincides with the local differential Galois group $G_{\infty}$ at $t=\infty$. For the exponential torus, we have
 \[
 \mathcal{T}_{\infty}=\begin{pmatrix}
 \lambda^{-1} & 0 & 0 & 0\\
 0 & \lambda^{-1} & 0 & 0\\
 0 & 0 & \lambda & 0\\
 0 & 0 & 0 & \lambda
 \end{pmatrix} ,
 \]
 where $\lambda\in\mathbb{C}^*$.

 As a result, we find that when $\alpha\in\mathbb{Q}$ the connected component $(G_1)^0$ of the differential Galois group $G_1$ coincides with $\mathcal{T}_{\infty}$. When $\alpha\notin\mathbb{Q}$ the group
 $(G_1)^0$ is generated by $\hat{M}_{\infty}$ and $\mathcal{T}_{\infty}$. Summarily, the group $(G_1)^0$ is defined as
 \[
 (G_1)^0=\left\{\begin{pmatrix}
 \mu^{-1} & 0 & 0 & 0\\
 0 & \mu^{-1} & 0 & 0\\
 0 & 0 & \mu & 0\\
 0 & 0 & 0 & \mu
 \end{pmatrix}, \ \mu\in\mathbb{C}^*\right\} ,
 \]
 which is an Abelian group.
 \end{proof}

 For the second normal variational equations $({\rm NVE})_2$ along the solution~\eqref{ps}, we obtain the
 non-homogeneous system
 \begin{gather*}
 \dot{x}_2 =
 \biggl(-1 + \frac{\alpha-1}{t}\biggr) x_2 + \biggl(\frac{2 \alpha^2}{t^3} - \frac{2 \alpha}{t^2}\biggr) y_2
 + \biggl(\frac{2 \alpha^2}{t^3} - \frac{2 \alpha}{t^2}\biggr) w_2
 + x_1^2 \\
 \hphantom{\dot{x}_2 =}{} +\biggl(\frac{4 \alpha}{t^2} - \frac{2}{t}\biggr) x_1 y_1 +
 \biggl(\frac{4 \alpha}{t^2} - \frac{2}{t}\biggr) w_1 z_1,\\
 \dot{y}_2 =
 \biggl(1 - \frac{ \alpha -1}{t} \biggr) y_2 -2 x_1 y_1 + \biggl(\frac{1}{t} - \frac{2 \alpha}{t^2}\biggr) y_1^2,\\
 \dot{z}_2 =
 \biggl(-1 + \frac{\alpha-1}{t} \biggr) z_2 +
 \biggl(\frac{2 \alpha^2}{t^3} - \frac{2 \alpha}{t^2}\biggr) w_2
 + \biggl(\frac{2 \alpha^2}{t^3} - \frac{2 \alpha}{t^2}\biggr) y_2 + z^2_1\\
 \hphantom{\dot{z}_2 =}{} + \biggl(\frac{4 \alpha}{t^2} - \frac{2}{t}\biggr) y_1 z_1 +
 \biggl(\frac{4 \alpha}{t^2} -\frac{2}{t}\biggr) w_1 z_1,\\
 \dot{w}_2 =
 \biggl(1 - \frac{\alpha-1}{t}\biggr) w_2 -
 \biggl(\frac{4 \alpha}{t^2} - \frac{2}{t}\biggr) y_1 w_1 - 2 w_1 z_1
 +\biggl(\frac{1}{t} - \frac{2 \alpha}{t^2}\biggr) w_1^2 .
 \end{gather*}
 Introducing more 9 new variables
 $x_1^2$, $x_1 y_1$, $x_1 w_1$, $w_1 z_1$, $w_1^2$, $w_1 y_1$, $y_1^2$, $z_1^2$, $z_1 y_1$ and their differential equations, we extend the $({\rm NVE})_2$ to the $({\rm LNVE})_2$ \cite{MRS}.
 The $({\rm LNVE})_2$ is a system of thirteenth order.
 The very high order of the $({\rm LNVE})_2$ make the problem of the description of its differential Galois group
 too complicated. Fortunately, it is not necessary to study the whole~$({\rm LNVE})_2$. If we find a subsystem
 of $({\rm LNVE})_2$, for which the connected component $G^0$ of the unit element of the corresponding
 differential Galois group is not Abelian and so is \smash{$(G_2)^0$}.

 For this reason, from here on we study the differential Galois group of the following fourth-order linear homogeneous system:
 \begin{gather}
 \dot{w}_2 =
 \biggl(1 - \frac{\alpha-1}{t}\biggr) w_2 - 2 p -
 \biggl(\frac{4 \alpha}{t^2} - \frac{2}{t}\biggr) q
 +\biggl(\frac{1}{t} - \frac{2 \alpha}{t^2}\biggr) v,\nonumber\\
 \dot{p} =
 \biggl(\frac{2 \alpha^2}{t^3} - \frac{2 \alpha}{t^2}\biggr) q
 + \biggl(\frac{2 \alpha^2}{t^3} - \frac{2 \alpha}{t^2}\biggr) v,\nonumber\\
 \dot{q} =
 2 \biggl(1- \frac{\alpha-1}{t}\biggr) q,\nonumber\\
 \dot{v} =
 2 \biggl(1- \frac{\alpha-1}{t}\biggr) v ,\label{nve2}
 \end{gather}
 where we have denoted $p:=w_1 z_1$, $q:=y_1 w_1$, $v:=w^2_1$. The system~\eqref{nve2} as the $({\rm NVE})_1$ has two singular
 points over $\mathbb{C}\mathbb{P}^1$: $t=0$ and $t=\infty$. The origin is a regular singularity, while $t=\infty$ is an irregular
 singularity of Poincar\'{e} rank 1.

 In order to determine the local differential Galois group at $t=\infty$ of the system~\eqref{nve2}, we make the
 change $t=1/\tau$. This transformation takes the system~\eqref{nve2} into the system
 \begin{gather}
 w'_2 =
 \biggl(\frac{\alpha-1}{\tau} - \frac{1}{\tau^2}\biggr) w_2 + \frac{2}{\tau^2} p
 + \biggl(4 \alpha -\frac{2}{\tau}\biggr) q + \biggl(2 \alpha - \frac{1}{\tau}\biggr) v,\nonumber\\
 p' =
 \bigl(2 \alpha - 2 \alpha^2 \tau\bigr) q + \bigl(2 \alpha - 2 \alpha^2 \tau\bigr) v,\nonumber\\
 q' =
 2 \biggl(\frac{\alpha-1}{\tau} - \frac{1}{\tau^2}\biggr) q,\nonumber\\
 v' =
 2 \biggl(\frac{\alpha-1}{\tau} - \frac{1}{\tau^2}\biggr) v ,\label{eq}
 \end{gather}
 where \smash{$'=\frac{{\rm d}}{{\rm d} \tau}$}. Now the origin is an irregular singularity of Poincar\'{e} rank 1 for the system~\eqref{eq}.
 We note that from here on we use the standard notations $(a)_n$ and \smash{$(a)^{(n)}$} for the falling and the rising factorials
 \begin{gather*}
 (a)_n=a (a-1) (a-2) \cdots (a-n+1), \qquad (a)_0=1,\\
 (a)^{(n)}=a (a+1) (a+2) \cdots (a+n-1), \qquad a^{(0)}=1,
 \end{gather*}
 respectively.

 \begin{Proposition}\label{for}
 The system~\eqref{eq} possesses an unique formal fundamental matrix solution at the origin in the form
 \[
 \hat{\Psi}(\tau)=\hat{H}(\tau) \tau^{\Lambda} \exp\biggl(\frac{Q}{\tau}\biggr) ,
 \]
 where the matrices $\Lambda$ and $Q$ are given by
 \begin{equation*}
 \Lambda={\rm diag} (\alpha-1, 0, 2 \alpha -2, 2 \alpha-2),\qquad
 Q={\rm diag}(1, 0, 2, 2) .
 \end{equation*}
 The matrix \smash{$\hat{H}(\tau)$} is defined as
 \[
 \hat{H}(\tau)=\begin{pmatrix}
 	 1 & 2+\hat{\varphi}(\tau) & 2 \tau & \hat{\phi}(\tau)\\
 	 0 & 1 & -\alpha \tau^2 & -\alpha \tau^2\\
 	 0 & 0 & 1 & 0\\
 	 0 & 0 & 0 & 1
 	\end{pmatrix} .
 \]
 The elements $\hat{\varphi}(\tau)$ and $\hat{\phi}(\tau)$ are defined as follows:
 \begin{enumerate}
 \itemsep=0pt 	
 \item[$1.$]
 If $\alpha\in\mathbb{N}$, the function $\hat{\varphi}(\tau)$ is the polynomial
 \begin{equation}
 \hat{\varphi}(\tau)= 2 (\alpha-1) \tau + 2 (\alpha-1) (\alpha-2) \tau^2 + \cdots + 2 (\alpha-1)! \tau^{\alpha-1} .\label{fs1}
 \end{equation}
 Otherwise, $\hat{\varphi}(\tau)$ is given by the following divergent power series:
 \[
 \hat{\varphi}(\tau)=2 \sum_{n=1}^{\infty} (\alpha-1)_n \tau^n .
 \]

 \item[$2.$]
 If $\alpha\in\mathbb{Z}_{\leq 0}$, the function $\hat{\phi}(\tau)$ is the polynomial
 \begin{equation}\
 \hat{\phi}(\tau)=\tau + \alpha \tau^2 + \alpha (\alpha+1) \tau^3 + \cdots + (-1)^{-\alpha} (-\alpha)! \tau^{-\alpha+1} .\label{fs2}
 \end{equation}
 Otherwise, $\hat{\phi}(\tau)$ is given by the following divergent power series:
 \[
 \hat{\phi}(\tau)=\sum_{n=0}^{\infty} \alpha^{(n)} \tau^{n+1} .
 \]
 \end{enumerate}
 \end{Proposition}

 \begin{proof}
 The formulas
 \[
 p(\tau)=C_3 - C_2 \alpha {\rm e}^{\frac{2}{\tau}} \tau^{2 \alpha} - C_1 \alpha {\rm e}^{\frac{2}{\tau}} \tau^{2 \alpha},\qquad
 q(\tau)=C_2 {\rm e}^{\frac{2}{\tau}} \tau^{2 (\alpha-1)}, \qquad
 v(\tau)=C_1 {\rm e}^{\frac{2}{\tau}} \tau^{2 (\alpha-1)} ,
 \]
 where $C_1$, $C_2$, $C_3$ are constant of integration, give the general solutions of the last three equations of the system~\eqref{eq}.
To build a local fundamental matrix solution $\hat{\Psi}(\tau)$ at the origin, we use that each column of such a matrix is a solution
 of the system~\eqref{eq}.
 Denote by $\hat{\Psi}_j(\tau)$, $j=1, 2, 3, 4$, the columns of the matrix $\hat{\Psi}(\tau)$. Then
 \[
 \hat{\Psi}_1(\tau)=\begin{pmatrix}
 \hat{w}_2^{(1)}(\tau)\\
 0 \\
 0 \\
 0
 \end{pmatrix} ,
 \]
 where \smash{$\hat{w}_2^{(1)}(\tau)$} is a solution of the equation
 \[
 w_2'=\biggl(\frac{\alpha-1}{\tau} - \frac{1}{\tau^2}\biggr) w_2 .
 \]
 We choose \smash{$\hat{w}_2^{(1)}(\tau)={\rm e}^{\frac{1}{\tau}} \tau^{\alpha-1}$}. For the second column $\hat{\Psi}_2(\tau)$, we have
 \[
 \hat{\Psi}_2(\tau)=\begin{pmatrix}
 \hat{w}_2^{(2)}(\tau)\\
 1 \\
 0 \\
 0
 \end{pmatrix} ,
 \]
 where \smash{$\hat{w}_2^{(2)}(\tau)$} is a formal solution of the equation
 \begin{equation}
 w_2'=\biggl(\frac{\alpha-1}{\tau} - \frac{1}{\tau^2}\biggr) w_2 + \frac{2}{\tau^2} .\label{eq1}
 \end{equation}
 The equation~\eqref{eq1} admits a formal solution near the origin in the form
 \[
 \hat{w}_2^{(2)}(\tau)=2 + 2 \sum_{n=1}^{\infty} (\alpha-1)_n \tau^n,
 \]
 which is a polynomial when $\alpha\in\mathbb{N}$. When $\alpha\notin\mathbb{N}$, the solution \smash{$\hat{w}_2^{(2)}(\tau)$} is a divergent power series.

 For the next column, we have
 \[
 \hat{\Psi}_3(\tau)=\begin{pmatrix}
 \hat{w}_2^{(3)}(\tau)\\[0.2ex]
 -\alpha {\rm e}^{\frac{2}{\tau}} \tau^{2 \alpha}\\[0.2ex]
 {\rm e}^{\frac{2}{\tau}} \tau^{2 (\alpha-1)}\\[0.2ex]
 0
 \end{pmatrix} ,
 \]
 where \smash{$\hat{w}_2^{(3)}(\tau)$} is a solution of the equation
 \[
 w_2'=\left(\frac{\alpha-1}{\tau} - \frac{1}{\tau^2}\right) w_2 + 2 \alpha {\rm e}^{\frac{2}{\tau}} \tau^{2 \alpha-2}
 -2 {\rm e}^{\frac{2}{\tau}} \tau^{2 \alpha -3} .
 \]
 We choose
 $\hat{w}_2^{(3)}(\tau)=2 {\rm e}^{\frac{2}{\tau}} \tau^{2 \alpha -1}$ .

 For the last column $\hat{\Psi}_4(\tau)$, we have
 \[
 \hat{\Psi}_4(\tau)=\begin{pmatrix}
 \hat{w}_2^{(4)}(\tau)\\[0.2ex]
 -\alpha {\rm e}^{\frac{2}{\alpha}} \tau^{2 \alpha}\\[0.2ex]
 0\\[0.2ex]
 {\rm e}^{\frac{2}{\tau}} \tau^{2 (\alpha-1)}
 	\end{pmatrix} ,
 \]
 where \smash{$\hat{w}_2^{(4)}(\tau)$} is a solution of the equation
 \begin{equation}
 w_2'=\biggl(\frac{\alpha-1}{\tau} - \frac{1}{\tau^2}\biggr) w_2
 - {\rm e}^{\frac{2}{\tau}} \tau^{2 \alpha -3} .\label{eq3}
 \end{equation}
 Looking for a solution of the equation~\eqref{eq3} in the form
 \smash{$\hat{w}_2^{(4)}(\tau)={\rm e}^{\frac{2}{\tau}} \tau^{2 \alpha-3} \hat{g}(\tau)$}, we find that~$\hat{g}(\tau)$ must satisfy the equation
 \begin{equation}\
 \tau^2 \hat{g}'=(1- (\alpha-2)) \hat{g} - \tau^2 .\label{eq4}
 \end{equation}
 The equation~\eqref{eq4} admits a formal solution near the origin in the form
 \[
 \hat{g}(\tau)=\tau \sum_{n=0}^{\infty} \alpha^{(n)} \tau^{n+1},
 \]
 which is a polynomial when $\alpha\in\mathbb{Z}_{\leq 0}$. Otherwise, $\hat{g}(\tau)$ is a divergent power series.

 Fitting together the so building columns $\hat{\Psi}_j(\tau)$, $j=1, 2, 3, 4$, we complete the proof.
 \end{proof}

 Now we can determine explicitly the formal invariants at the origin of the system~\eqref{eq}.

 \begin{Proposition}
 With respect to the formal fundamental matrix solution $\hat{\Psi}(\tau)$ introduced by Proposition $\ref{for}$, the
 exponential torus $\mathcal{T}$ and the formal monodromy $\hat{M}_0$ at the origin of the system~\eqref{eq} are given by
 \[
 \mathcal{T}=\left\{\begin{pmatrix}
 \lambda & 0 & 0 & 0\\
 0 & 1 & 0 & 0\\
 0 & 0 & \lambda^2 & 0\\
 0 & 0 & 0 & \lambda^2
 \end{pmatrix},\ \lambda\in\mathbb{C}^*\right\},\qquad
 \hat{M}_0={\rm e}^{2 \pi {\rm i} \Lambda}=\begin{pmatrix}
 {\rm e}^{2 \pi {\rm i} \alpha} & 0 & 0 & 0\\
 0 & 1 & 0 & 0\\
 0 & 0 & {\rm e}^{4 \pi {\rm i} \alpha} & 0\\
 0 & 0 & 0 & {\rm e}^{4 \pi {\rm i} \alpha}
 \end{pmatrix} .
 \]
 \end{Proposition}

 The application of Definition \ref{st} to the divergent power series $\hat{\varphi}(\tau)$ and $\hat{\phi}(\tau)$ gives
 the following sets of admissible singular directions:
 \[
 \Theta_2=\{ \theta=\arg (0-1)=\arg (0-2)=\pi \}
 \]
 for the series $\hat{\varphi}(\tau)$ and
 \[
 \Theta_3=\{ \theta=\arg (2-1)=\arg (2-0)=0 \}
 \]
 for the series $\hat{\phi}(\tau)$.

 In the next lemma, we compute the 1-sums of the divergent power series $\hat{\varphi}(\tau)$ and
 $\hat{\phi}(\tau)$. These 1-sums illustrate explicitly the dependence of the power series \smash{$\hat{\varphi}(\tau)$} and
 \smash{$\hat{\phi}(\tau)$} on the suggested admissible singular direction.
 Denote $a={\rm Re} (\alpha)$ and by $[|a|]$, $[a]$ the integer parts of the real numbers $|a|$ and $a$, respectively.
 \begin{Lemma}\label{sum}
 Under the above notations, we have
 \begin{enumerate} 	\itemsep=0pt
 \item[$1.$]
 	Assume that $\alpha\notin\mathbb{N}$. Then for every direction $\theta \neq \pi$ the function
 	 \[
 	 \varphi_{\theta}(\tau)=2 (\alpha-1)
 	 \int_0^{+\infty {\rm e}^{{\rm i} \theta}}
 	 (1+\zeta)^{\alpha-2} {\rm e}^{-\frac{\zeta}{\tau}} {\rm d} \zeta
 	 \]
 	 defines the $1$-sum of the power series $\hat{\varphi}(\tau)$ in such a direction.

 \item[$2.$]
 Assume that $\alpha\notin\mathbb{Z}_{\leq 0}$. Then for every direction $\theta \neq 0$ the function
 \[
 \phi_{\theta}(\tau)=\int_0^{+\infty {\rm e}^{{\rm i} \theta}}
 (1-\zeta)^{-\alpha} {\rm e}^{-\frac{\zeta}{\tau}} {\rm d} \zeta
 \]
 defines the $1$-sum of the power series $\hat{\phi}(\tau)$ in such a direction.

 \end{enumerate}

 When ${\rm Re} (\alpha) \leq 2$ $($resp.\ ${\rm Re} (\alpha) \geq 0)$ the function $\varphi_{\theta}(\tau)$
 $($resp.\ $\phi_{\theta}(\theta))$ is a holomorphic function in the open unlimited disc
 \[
 \mathcal{D}_{\theta}=\biggl\{ \tau\in\mathbb{C} \mid {\rm Re} \biggl(\frac{{\rm e}^{{\rm i} \theta}}{\tau}\biggr) > 0 \biggr\}
 \]
 for any direction $\theta \neq \pi$ $($resp.\ $\theta\neq 0)$.
 Otherwise, they are holomorphic functions in the open bounded disc
 \[
 \mathcal{D}_{\theta}(1)=\biggl\{ \tau\in\mathbb{C} \mid {\rm Re} \biggl(\frac{{\rm e}^{{\rm i} \theta}}{\tau}\biggr) > 1 \biggr\} .
 \]
 \end{Lemma}

 \begin{proof}
 Consider the divergent power series $\hat{\varphi}(\tau)$ and $\hat{\phi}(\tau)$ defined by Proposition \ref{for}. For the corresponding
 formal Borel transforms, we obtain the power series
 \begin{gather*}
 \hat{\mathcal{B}}_1 \hat{\varphi}(\zeta)
 =
 2 \sum_{n=0}^{\infty} (\alpha-1)_{n+1} \frac{\zeta^n}{n!}=
 2 (\alpha-1) \sum_{n=0}^{\infty} (\alpha-2)_n \frac{\zeta^n}{n!},\qquad
 \hat{\mathcal{B}}_1 \hat{\phi}(\zeta)
 =
 \sum_{n=0}^{\infty} \alpha^{(n)} \frac{\zeta^n}{n!} .
 \end{gather*}
 Both of the series $\hat{\mathcal{B}}_1 \hat{\varphi}(\zeta)$ and $\hat{\mathcal{B}}_1 \hat{\phi}(\zeta)$ are convergent near the origin
 of the Borel $\zeta$-plane with finite radiuses of convergence. Therefore, both of the divergent series $\hat{\varphi}(\tau)$
 and $\hat{\phi}(\tau)$ are Gevrey series of order 1. The functions
 \[
 \varphi(\zeta)=2 (\alpha-1) (1+\zeta)^{\alpha-2},\qquad
 \phi(\zeta)=(1-\zeta)^{-\alpha}
 \]
 present the sums of the power series $\hat{\mathcal{B}}_1 \hat{\varphi}(\zeta)$ and $\hat{\mathcal{B}}_1 \hat{\phi}(\zeta)$, respectively.

 Consider the function $(1+\zeta)^{\alpha-2}$. We have that $\bigl|(1+\zeta)^{\alpha-2}\bigr|=A |1+\zeta|^{{\rm Re}(\alpha) -2}$, where
 $A={\rm e}^{-{\rm Im} (\alpha) \arg (1+\zeta)}$. Let $\theta=\arg (\zeta)$. If ${\rm Re} (\alpha) -2 \leq 0$, then
 \[
 \frac{A}{|1+\zeta|^{2 - {\rm Re}(\alpha)}} \leq A
 \]
 when $\cos \theta \geq 0$, while
 \[
 \frac{A}{|1+\zeta|^{2 - {\rm Re}(\alpha)}} \leq \frac{A}{|\sin \theta|^{2 - {\rm Re}(\alpha)}}
 \]
 when $\cos \alpha < 0$. If ${\rm Re} (\alpha) - 2 >0$, then
 \[
 A |1+\zeta|^{{\rm Re}(\alpha) -2} \leq B {\rm e}^{|\zeta|}
 \]
 for an appropriate constant $B > 0$. Therefore, the function $\varphi(\zeta)$ is of exponential size at most~1 at~$\infty$
 along any direction $\theta \neq \pi$ from $0$ to $+\infty {\rm e}^{{\rm i} \theta}$. Moreover, the function $\varphi(\zeta)$ is
 continued analytically along any such a direction. Then the corresponding Laplace transform $(\mathcal{L}_{\theta} \varphi)(\tau)$
 is well defined and gives the 1-sum of the divergent power series $\hat{\varphi}(\tau)$ in such a direction. If we denote
 by $\varphi_{\theta}(\tau)$ this 1-sum, we have that
 \[
 \varphi_{\theta}(\tau)=2 (\alpha-1)
 \int_0^{+\infty {\rm e}^{{\rm i} \theta}}
 (1+\zeta)^{\alpha-2} {\rm e}^{-\frac{\zeta}{\tau}} {\rm d} \zeta
 \]
 for every $\theta \neq \pi$. From the above estimates, it follows that when ${\rm Re} (\alpha) -2 \leq 0$, the 1-sum
 $\varphi_{\theta}(\tau)$ is a holomorphic function in the open unlimited disc
 \begin{equation}
 \mathcal{D}_{\theta}=\biggl\{ \tau\in\mathbb{C} \mid {\rm Re} \biggl(\frac{{\rm e}^{{\rm i} \theta}}{\tau}\biggr) > 0 \biggr\} ,\label{d}
 \end{equation}
 whose opening is $\pi$.
 When ${\rm Re} (\alpha) - 2 > 0$, the 1-sum is a holomorphic function in the open bounded disc
 \begin{equation}
 \mathcal{D}_{\theta}(1)=\biggl\{ \tau\in\mathbb{C} \mid {\rm Re} \biggl(\frac{{\rm e}^{{\rm i} \theta}}{\tau}\biggr) > 1 \biggr\}\label{d1}
 \end{equation}
 with opening $ < \pi$.

 Using similar arguments, we find that for any direction $\theta \neq 0$ the Laplace transform
 \[
 \phi_{\theta}(\tau)=\int_0^{+\infty {\rm e}^{{\rm i} \theta}}
 (1-\zeta)^{-\alpha} {\rm e}^{-\frac{\zeta}{\tau}} {\rm d} \zeta
 \]
 defines the 1-sum of the divergent power series $\hat{\phi}(\tau)$ in such a direction. When ${\rm Re} (\alpha) \geq 0$,
 the 1-sum $\phi_{\theta}(\tau)$ is a holomorphic function in the disc $\mathcal{D}_{\theta}$ from~\eqref{d}. Otherwise, the
 1-sum is a~holomorphic function in the disc $\mathcal{D}_{\theta}(1)$ from~\eqref{d1}.
 \end{proof}

 \begin{Remark}\label{sum1}
 Let $I =(-\pi, \pi) \subset \mathbb{R}$ and $J=(0, 2 \pi) \subset \mathbb{R}$. When we move the direction $\theta\in I$, the
 holomorphic functions $\varphi_{\theta}(\tau)$ glue together analytically and define a holomorphic function~$\tilde{\varphi}(\tau)$ on a sector $\widetilde{\mathcal{D}}_1$ with opening
 $3 \pi$, $-\frac{3 \pi}{2} < \arg (\tau) < \frac{3 \pi}{2}$ when ${\rm Re}(\alpha) \leq 2$ or on a sector
 \[
 \widetilde{\mathcal{D}}_1(1)=\bigcup _{\theta\in I} \mathcal{D}_{\theta}(1)
 \]
 with opening $> \pi$ when ${\rm Re}(\alpha) > 2$.
 Similarly, when we move the direction $\theta\in J$, the
 holomorphic functions $\phi_{\theta}(\tau)$ glue together analytically and define a holomorphic function
 $\tilde{\phi}(\tau)$ on a sector~$\widetilde{\mathcal{D}}_2$ with opening
 $3 \pi$, $-\frac{\pi}{2} < \arg (\tau) < \frac{5 \pi}{2}$ when ${\rm Re}(\alpha) \geq 0$ or on a sector
 \[
 \widetilde{\mathcal{D}}_2(1)=\bigcup _{\theta\in J} \mathcal{D}_{\theta}(1)
 \]
 with opening $> \pi$ when ${\rm Re}(\alpha) < 0$. On these sectors, the functions $\tilde{\varphi}(\tau)$ and
 $\tilde{\phi}(\tau)$ are asymptotic to the power series $\hat{\varphi}(\tau)$ and \smash{$\hat{\phi}(\tau)$}, respectively,
 in Gevrey 1-sense and define the 1-sums of these power series there. Their restrictions on $\mathbb{C}^*$ are multivalued
 functions. In any direction~${\theta\neq \pi}$ (resp.\ $\theta\neq 0$) the function $\tilde{\varphi}(\tau)$
 \big(resp.\ \smash{$\tilde{\phi}(\tau)$}\big) has only one value which coincides with the function~$\varphi_{\theta}(\tau)$
 (resp.\ $\phi_{\theta}(\tau)$) defined by Lemma \ref{sum}. Near the singular direction $\theta=\pi$
 (resp.\ $\theta=0$) the function $\tilde{\varphi}(\tau)$ \big(resp.\ $\tilde{\phi}(\tau)$\big) has two different values:
 $\varphi^{+}_{\pi}(\tau)=\varphi_{\pi+\varepsilon}(\tau)$
 \big(resp.~${\phi^{+}_{0}(\tau)=\phi_{0+\varepsilon}(\tau)}$\big) and $\varphi^{-}_{\pi}(\tau)=\varphi_{\pi-\varepsilon}(\tau)$
 (resp.\ $\phi^{-}_{0}(\tau)=\phi_{0-\varepsilon}(\tau)$) for a small number $\varepsilon>0$.

 \end{Remark}

 Replacing the divergent power series entries of the matrix \smash{$\hat{H}(\tau)$} with their 1-sums, we get an actual fundamental
 matrix solution at the origin of the system~\eqref{eq}. More precisely,
 denote $F(\tau)=\tau^{\Lambda} \exp\bigl(\frac{Q}{\tau}\bigr)$.

 \begin{Proposition}\label{ac}
 For every non-singular direction $\theta$, the system~\eqref{eq} possesses an unique actual fundamental matrix solution
 at the origin in the form
 \begin{equation}
 \Psi_{\theta}(\tau)=H_{\theta}(\tau) F_{\theta}(\tau) ,\label{afms}
 \end{equation}
 where $F_{\theta}(\tau)$ is the branch of the matrix $F(\tau)$ for $\theta=\arg (\tau)$. The matrix $H_{\theta}(\tau)$
 is given by
 \[
 H_{\theta}(\tau)=\begin{pmatrix}
 1 & 2 + h_{12}(\tau) & 2 \tau & h_{14}(\tau)\\
 0 & 1 & -\alpha \tau^2 & -\alpha \tau^2\\
 0 & 0 & 1 & 0\\
 0 & 0 & 0 & 1
 \end{pmatrix} .
 \]
 The entries $h_{12}(\tau)$ and $h_{14}(\tau)$ of the matrix $H_{\theta}(\tau)$ are defined as follows:
 \begin{enumerate}
 	\itemsep=0pt
 \item[$1.$]
 If $\alpha\in\mathbb{N}$, then $h_{14}(\tau)=\phi_{\theta}(\tau)$, where $\phi_{\theta}(\tau)$ is defined by Lemma~$\ref{sum}$
 and extended by Remark~$\ref{sum1}$. The element $h_{12}(\tau)$ coincides with the function $\hat{\varphi}(\tau)$ from
 \eqref{fs1}.

 \item[$2.$]
 If $\alpha\in\mathbb{Z}_{\leq 0}$, then $h_{12}(\tau)=\varphi_{\theta}(\tau)$, where $\varphi_{\theta}(\tau)$ is defined
 by Lemma~$\ref{sum}$ and extended by Remark~$\ref{sum1}$. The element $h_{14}(\tau)$ coincides with the function $\hat{\phi}(\tau)$
 from~\eqref{fs2}.

 \item[$3.$]
 If $\alpha\notin\mathbb{Z}$, then $h_{12}(\tau)=\varphi_{\theta}(\tau)$, $h_{14}(\tau)=\phi_{\theta}(\tau)$,
 where $\varphi_{\theta}(\tau)$ and $\phi_{\theta}(\tau)$ are defined by Lemma~$\ref{sum}$ and extended by Remark~$\ref{sum1}$.
\end{enumerate}
 	
 Near the singular direction $\theta=0$ or $\theta=\pi$, the system~\eqref{eq} possesses two different actual
 fundamental matrix solutions at the origin in the form
 \[
 \Psi_{0}^+(\tau)=\Psi_{0+\varepsilon}(\tau) \qquad {\rm and} \qquad \Psi_0^{-}(\tau)=\Psi_{0-\varepsilon}(\tau)
 \] 	
 or
 \[
 \Psi_{\pi}^+(\tau)=\Psi_{\pi+\varepsilon}(\tau) \qquad {\rm and} \qquad \Psi_{\pi}^{-}(\tau)=\Psi_{\pi-\varepsilon}(\tau),
 \] 	
 where $\Psi_{0\pm \varepsilon}(\tau)$ and $\Psi_{\pi \pm \varepsilon}(\tau)$ are defined by~\eqref{afms} for a small number
 $\varepsilon > 0$.
 \end{Proposition}

 Now we can compute the analytic invariants at the origin of the system~\eqref{eq}.

 \begin{Theorem}
 With respect to the actual fundamental matrix solution at the origin, defined by Proposition $\ref{ac}$, the system~\eqref{eq} has a
 Stokes matrix $St_{\pi}$ in the form
 \[
 St_{\pi}=\begin{pmatrix}
 1 & \mu_1 & 0 & 0\\
 0 & 1 & 0 & 0\\
 0 & 0 & 1 & 0\\
 0 & 0 & 0 & 1
 \end{pmatrix} .
 \]
 The multiplier $\mu_1$ is defined as follows:
 \begin{enumerate}
 \itemsep=0pt 	
 \item[$1.$]
 If ${\rm Re} (\alpha-1) > 0$, then
 \[
 \mu_1= 4 {\rm i} (\alpha-1) (-1)^{\alpha-1} \sin((2-\alpha) \pi) \Gamma(\alpha-1) .
 \]
 	
 \item[$2.$]
 If ${\rm Re} (\alpha-1) \leq 0$ but $\alpha\notin\mathbb{Z}_{\leq 0}$, then
 \[
 \mu_1=
 \frac{4 \pi {\rm i} (\alpha-1) (-1)^{\alpha-1}}{\Gamma(2-\alpha)} .
 \]
 	
 \item[$3.$]
 If $\alpha\in\mathbb{Z}_{\leq 0}$, then
 \[
 \mu_1=\frac{4 \pi {\rm i} (-1)^{-\alpha}}{(-\alpha)!} .
 \]
 \end{enumerate}

 Similarly, with respect to the actual fundamental matrix solution at the origin, defined by Proposition $\ref{ac}$, the system~\eqref{eq} has a
 Stokes matrix $St_0$ in the form
 \[
 St_0=\begin{pmatrix}
 1 & 0 & 0 & \mu_2\\
 0 & 1 & 0 & 0\\
 0 & 0 & 1 & 0\\
 0 & 0 & 0 & 1
 \end{pmatrix} .
 \]
 The multiplier $\mu_2$ is defined as follows:
 \begin{enumerate}
 \itemsep=0pt
 	\item[$1.$]
 	If ${\rm Re} (1-\alpha) > 0$, then
 	\[
 	\mu_2=
 	 -2 {\rm i} \sin((1-\alpha) \pi) \Gamma(1-\alpha) .
 	\]
 	
 	\item[$2.$]
 	If ${\rm Re} (1-\alpha) \leq 0$ but $\alpha\notin\mathbb{N}$, then
 	\[
 	\mu_2=-\frac{2 \pi {\rm i}}{\Gamma(\alpha)} .
 	\]
 	
 	\item[$3.$]
 	If $\alpha\in\mathbb{N}$, then
 	\[
 	\mu_2=-\frac{2 \pi {\rm i}}{(\alpha-1)!} .
 	\]
 \end{enumerate}
 \end{Theorem}

 \begin{proof}
 From the Definition \ref{Sto}, it follows that the multiplier $\mu_1$ is computed by comparing the solutions
 $\varphi^{-}_{\pi}(\tau)$ and $\varphi^{+}_{\pi}(\tau)$. Denote
 \[
 J_1= \varphi^{-}_{\pi}(\tau) - \varphi^{+}_{\pi}(\tau) .
 \]
 Then
 \[
 J_1=
 2 (\alpha-1) \int_\gamma (1+\zeta)^{\alpha-2} {\rm e}^{-\zeta/\tau} {\rm d} \zeta\qquad {\rm for}\
 \frac{\pi}{2} < \arg (\tau) < \frac{3 \pi}{2} ,
 \]
 where $\gamma=(\pi-\varepsilon) - (\pi+\varepsilon)$ for a small number $\varepsilon > 0$. Assume that $\alpha\notin\mathbb{Z}_{\leq 0}$. Then
 without changing the integral, we can deform the path $\gamma$ into a Hankel type contour $\gamma_1$ winding around
 the branch cut on $\mathbb{R}^{-}$ of the function $(1+\zeta)^{\alpha-2}$, starting on $-\infty$, encircling $-1$ in the positive sense and returning to $-\infty$.
 Then $J_1$ becomes
 \[
 J_1=
 2 (\alpha-1) \int_{\gamma_1}
 (1+\zeta)^{\alpha-2} {\rm e}^{-\zeta/\tau} {\rm d} \zeta .
 \]
 The transformation $1+\zeta=u$ takes the contour $\gamma_1$ into a Hankel type contour $\gamma_2$ going along the branch cut
 on $\mathbb{R}^{-}$ of the function $u^{\alpha-2}$,
 starting on $-\infty$, encircling $0$ in the positive sense and backing to $-\infty$. Then we have
 \[
 J_1=
 2 (\alpha-1) {\rm e}^{1/\tau}\int_{\gamma_2}
 u^{\alpha-2} {\rm e}^{-u/\tau} {\rm d} u .
 \]
 Now the change $u/\tau=-\eta$ takes the contour $\gamma_2$ into itself and we find that
 \[
 J_1=
 2 (\alpha-1) (-1)^{\alpha-1} \tau^{\alpha-1} {\rm e}^{1/\tau}\int_{\gamma_2}
 \eta^{\alpha-2} {\rm e}^{\eta} {\rm d} \eta .
 \]
 To obtain the formula for $\mu_1$, we use the well-known Hankel's representation of the Gamma function $\Gamma(1-\alpha)$
 and reciprocal Gamma function $1/\Gamma(\alpha)$ when $\alpha\neq 0, -1, -2, \dots$
 as a contour integral (see \cite[1.6\,(1) and 1.6\,(2)]{BE})
 \begin{equation*}
 \int_{\gamma_2} \eta^{-\alpha} {\rm e}^{\eta} {\rm d} \eta=2 {\rm i} \sin (\alpha \pi) \Gamma(1-\alpha),\qquad
 \frac{1}{2 \pi {\rm i}} \int_{\gamma_2} \eta^{-\alpha} {\rm e}^{\eta} {\rm d} \eta=\frac{1}{\Gamma(\alpha)} .
 \end{equation*}
 Hence
 \[
 \varphi^{-}_{\pi}(\tau) - \varphi^{+}_{\pi}(\tau)=
 4 {\rm i} (\alpha-1) (-1)^{\alpha-1} \sin((2-\alpha) \pi) \Gamma(\alpha-1)\tau^{\alpha-1} {\rm e}^{1/\tau}
 \]
 when ${\rm Re} (\alpha-1) > 0$, and
 \[
 \varphi^{-}_{\pi}(\tau) - \varphi^{+}_{\pi}(\tau)=
 \frac{4 \pi {\rm i} (\alpha-1) (-1)^{\alpha-1}}{\Gamma(2-\alpha)} \tau^{\alpha-1} {\rm e}^{1/\tau}
 \]
 when ${\rm Re} (\alpha-1) \leq 0$ but $\alpha\notin\mathbb{Z}_{\leq 0}$.

 Assume now that $\alpha\in\mathbb{Z}_{\leq 0}$. Then from the Cauchy's differential formula it follows that
 \[
 \varphi^{-}_{\pi}(\tau) - \varphi^{+}_{\pi}(\tau)
 =
 \frac{4 (\alpha-1) \pi {\rm i}}{(1-\alpha)!}
 \bigl[ D_{\zeta}^{1-\alpha} \bigl({\rm e}^{-\frac{\zeta}{\tau}}\bigr)_{| \zeta=-1}\bigr]
 =\frac{4 \pi {\rm i} (-1)^{-\alpha}}{(-\alpha)!} \tau^{\alpha-1} {\rm e}^{\frac{1}{\tau}} ,
 \]
 where \smash{$D^n_{\zeta}=\frac{d^{n}}{{\rm d} \zeta^n}$}.

 In a similar manner comparing the solutions $ \tau^{2 \alpha-2} {\rm e}^{2/\tau} \phi_0^{-}(\tau)$ and $ \tau^{2 \alpha-2} {\rm e}^{2/\tau} \phi_0^{+}(\tau)$,
 one can derive the multiplier $\mu_2$.
\end{proof}

 Now we can describe the local differential Galois group $G$ at the origin of the system~\eqref{eq}.

 \begin{Theorem}\label{gal-inf}
 With respect to the formal and actual fundamental matrix solutions, given by Propositions $\ref{for}$ and~$\ref{ac}$, the connected
 component $G^0$ of the unit element of the local differential Galois group $G$ at the origin of the system~\eqref{eq} is
 defined as follows:
 \begin{enumerate}
 \itemsep=0pt 	
 \item[$1.$]
 	If $\alpha\in\mathbb{N}$, then
 	\[
 	 G^0=\left\{\begin{pmatrix}
 	 \lambda & 0 & 0 & \mu\\
 	 0 & 1 & 0 & 0\\
 	 0 & 0 & \lambda^2 & 0\\
 	 0 & 0 & 0 & \lambda^2
 	 \end{pmatrix}, \ 	 \lambda\in\mathbb{C}^*,\, \mu\in\mathbb{C}\right\} .
 	\]
 	
 \item[$2.$]
 If $\alpha\in\mathbb{Z}_{\leq 0}$, then
 		\[
 		G^0=\left\{\begin{pmatrix}
 		\lambda & \mu & 0 & 0\\
 		0 & 1 & 0 & 0\\
 		0 & 0 & \lambda^2 & 0\\
 		0 & 0 & 0 & \lambda^2
 		\end{pmatrix}, \
 		\lambda\in\mathbb{C}^*,\, \mu\in\mathbb{C}\right\} .
 		\]
 		
 \item[$3.$]
 If $\alpha\notin\mathbb{Z}$, then
 	\[
 	G^0=\left\{\begin{pmatrix}
 	\lambda & \mu & 0 & \nu\\
 	0 & 1 & 0 & 0\\
 	0 & 0 & \lambda^2 & 0\\
 	0 & 0 & 0 & \lambda^2
 	\end{pmatrix}, \
 	\lambda\in\mathbb{C}^*,\, \mu, \nu\in\mathbb{C}\right\} .
 	\] 		
 	\end{enumerate}
 \end{Theorem}

 \begin{proof}
 From Theorem \ref{Ra}, it follows that the local differential Galois group $G$ at $\tau=0$ of the system~\eqref{eq}
 is the Zariski closure of the group generated by the formal differential Galois group
 and the Stokes matrices $St_{\pi}$ and $St_0$. Since when $\alpha\in\mathbb{Z}$ the formal monodromy $\hat{M}_0$ is equal
 to the identity matrix $I_4$ then the formal differential Galois group coincides with the exponential torus $\mathcal{T}$.
 Therefore, in this case $G$ is generated by the exponential torus and the
 Stokes matrices $St_{\pi}$ and $St_0$. Since when $\alpha\in\mathbb{N}$ the Stokes matrix $St_{\pi}$ is equal to $I_4$ the
 Galois group $G$ is generated by the exponential torus $\mathcal{T}$ and the Stokes matrix $St_0$. Hence for $\alpha\in\mathbb{N}$
 the differential Galois group coincides with its connected component $G^0$ of the unit element and
 \[
 G=G^0=\left\{\begin{pmatrix}
 \lambda & 0 & 0 & \mu\\
 0 & 1 & 0 & 0\\
 0 & 0 & \lambda^2 & 0\\
 0 & 0 & 0 & \lambda^2
 \end{pmatrix}, \ \lambda\in\mathbb{C}^*,\, \mu\in\mathbb{C}\right\} .
 \]
Similarly, when $\alpha\in\mathbb{Z}_{\leq 0}$ we have that $G$
 is generated topologically by $\mathcal{T}$ and $St_{\pi}$ and
 \[
 G=G^0=\left\{\begin{pmatrix}
 \lambda & \mu & 0 & 0\\
 0 & 1 & 0 & 0\\
 0 & 0 & \lambda^2 & 0\\
 0 & 0 & 0 & \lambda^2
 \end{pmatrix}, \ \lambda\in\mathbb{C}^*,\, \mu\in\mathbb{C}\right\} .
 \]

 When $\alpha\in\mathbb{Q}$ but $\alpha\notin\mathbb{Z}$, the formal differential Galois group is not connected since the group generated
 by $\hat{M}_0$ is not connected. However, in this case the connected component of the unit element of the formal differential
 Galois group coincides with $\mathcal{T}$. Therefore, in this case $G$ does not coincides with $G^0$ but $G^0$ is generated by $\mathcal{T}$ and
 the Stokes matrices $St_{\pi}$ and $St_0$. Hence
 \[
 G^0=\left\{\begin{pmatrix}
 \lambda & \mu & 0 & \nu\\
 0 & 1 & 0 & 0\\
 0 & 0 & \lambda^2 & 0\\
 0 & 0 & 0 & \lambda^2
 \end{pmatrix}, \ \lambda\in\mathbb{C}^*,\, \mu, \nu\in\mathbb{C}\right\} .
 \]

 In the last case when $\alpha\notin\mathbb{Q}$, we have that
 \[
 G=G^0=\left\{\begin{pmatrix}
 \lambda & \mu & 0 & \nu\\
 0 & 1 & 0 & 0\\
 0 & 0 & \lambda^2 & 0\\
 0 & 0 & 0 & \lambda^2
 \end{pmatrix}, \ \lambda\in\mathbb{C}^*,\, \mu, \nu\in\mathbb{C}\right\} .
 \]
 This ends the proof.
 \end{proof}

 Directly from Theorem \ref{gal-inf}, we obtain the following important result

 \begin{Theorem}\label{g0}
 The connected component of the unit element of the differential Galois group of the system~\eqref{nve2}
 is not Abelian.
 \end{Theorem}

 \begin{proof}
 If we prove that the connected component $G^0$ of the unit element of the local differential Galois group at the origin of the system~\eqref{eq} is not Abelian
 group, then we will have that the connected component of the unit element of the differential Galois group of the system~\eqref{nve2} will be not Abelian too.
 To prove that $G^0$, defined by Theorem \ref{gal-inf}, is not an Abelian group it is enough to show that the matrices{\samepage
 \[
 T_{\lambda}=\begin{pmatrix}
 	\lambda & 0 & 0 & 0\\
 	0 & 1 & 0 & 0\\
 	0 & 0 & \lambda^2 & 0\\
 	0 & 0 & 0 & \lambda^2
 \end{pmatrix} \qquad {\rm and} \qquad
 S_{\mu, \nu}=\begin{pmatrix}
 	1 & \mu & 0 & \nu\\
 	0 & 1 & 0 & 0\\
 	0 & 0 & 1 & 0\\
 	0 & 0 & 0 & 1
 \end{pmatrix}
 \]
 do not commute.}

 When $\mu$ and $\nu$ do not vanish together, the commutator between $S_{\mu, \nu}$ and $T_{\lambda}$
 \[
 S_{\mu, \nu} T_{\lambda} S^{-1}_{\mu, \nu} T^{-1}_{\lambda}=
 \begin{pmatrix}
 1 & \mu (1-\lambda) & 0 & \nu \bigl(1- \lambda^{-1}\bigr)\\
 0 & 1 & 0 & 0\\
 0 & 0 & 1 & 0\\
 0 & 0 & 0 & 1
 \end{pmatrix}
 \]
 is not identically equal to the identity matrix. The condition $\lambda=1$ implies than for every~${\sigma\in G}$ we have that \smash{$\sigma\bigl({\rm e}^{\frac{1}{\tau}}\bigr)={\rm e}^{\frac{1}{\tau}}$}, i.e.,
 \smash{${\rm e}^{\frac{1}{\tau}}\in\mathbb{C}(\tau)$}, which is a contradiction. Thus the connected~component of the unit element
 of the differential Galois group of the system~\eqref{nve2} is not Abelian.
 \end{proof}

 As a consequence of Theorem \ref{g0}, we have the following.

 \begin{Corollary}\label{go}
 The connected component \smash{$(G_2)^0$} of the unit element of the differential Galois group of the
 $({\rm LNVE})_2$ is not Abelian.
 \end{Corollary}

\begin{proof}
We always can put the independent variables in $({\rm LNVE})_2$ in such an order that the variables
 $w_2$, $w_1 z_1$, $y_1 w_1$, $w_1^2$ stay in a block. For example, let the system~\eqref{nve2} forms the tail of
 the~$({\rm LNVE})_2$. Then after the transformation $t=1/\tau$ the $({\rm LNVE})_2$ admit such
 formal~\smash{$\widehat{\Psi}_{({\rm LNVE}_2)}(\tau)$} and actual \smash{$\Psi^{\theta}_{({\rm LNVE})_2}(\tau)$} fundamental matrix solutions
 at $\tau=0$ which contain the matrices~$\hat{\Psi}(\tau)$ and~$\Psi_{\theta}(\tau)$ from Propositions \ref{for} and~\ref{ac},
 respectively, as right-hand lower corner blocks.
 The differential Galois group $G_2$ of the $({\rm LNVE})_2$ is a subgroup of~${\rm GL}_{13}(\mathbb{C})$, so is \smash{$(G_2)^0$}.
 With respect to the fundamental matrix solutions
 \smash{$\widehat{\Psi}_{({\rm LNVE}_2)}(\tau)$} and~\smash{$\Psi^{\theta}_{({\rm LNVE})_2}(\tau)$} the connected component
 of the unit element \smash{$(G_2)^0$} of $G_2$ is not Abelian since it has a proper
 subgroup, which is not Abelian.
\end{proof}

Combining Corollary \ref{go} with the Morales--Ramis--Sim\'{o} theory, we establish the main result of this section.

 \begin{Theorem}\label{main-s}
 Assume that $\alpha_1=\alpha_2=\alpha_3=0$, $\alpha_4=1$, $\alpha_0=-\alpha_5$, where $\alpha_5$ is arbitrary. Then the
 Sasano system~\eqref{s} is not integrable in the Liouville--Arnold sense by rational first integrals.
 \end{Theorem}

 \section{B\"{a}cklund transformations and generalization}\label{section4}

 In this section with the aid of the B\"{a}cklund transformations of the Sasano system~\eqref{s}, we extend the result of Theorem \ref{main-s}
 to the entire orbit of the parameters $\alpha_j$, $j=0, \dots, 5$, and establish the main results of this paper.

 Denote $(*): =(x, y, z, w, t ; \alpha_0, \alpha_1, \alpha_2, \alpha_3, \alpha_4, \alpha_5)$.
 The action of the generators of the exten\-ded~affine Weyl group \smash{$\widetilde{W}\bigl(D^{(1)}_5\bigr)$} on $(*)$ is defined as follows (see \cite{Sa}):
 \begin{gather}
 s_0(*) =
 \biggl(x + \frac{\alpha_0}{y+t}, y, z, w, t ; -\alpha_0, \alpha_1, \alpha_2+\alpha_0, \alpha_3, \alpha_4, \alpha_5\biggr),\nonumber\\
 s_1(*) =
 \biggl(x + \frac{\alpha_1}{y}, y, z, w, t ; \alpha_0, -\alpha_1, \alpha_2+\alpha_1, \alpha_3, \alpha_4, \alpha_5\biggr),\nonumber\\
 s_2(*) =
 \biggl(x, y - \frac{\alpha_2}{x-z}, z, w + \frac{\alpha_2}{x-z}, t ; \alpha_0+\alpha_2, \alpha_1+\alpha_2, -\alpha_2,
 \alpha_3+\alpha_2, \alpha_4, \alpha_5\biggr),\nonumber\\
 s_3(*) =
 \biggl(x, y, z + \frac{\alpha_3}{w}, w, t ; \alpha_0, \alpha_1, \alpha_2+\alpha_3, -\alpha_3, \alpha_4+\alpha_3, \alpha_5+\alpha_3\biggr),
 \nonumber
 \\
 s_4 (*) =
 \biggl(x, y, z, w - \frac{\alpha_4}{z-1}, t ; \alpha_0, \alpha_1, \alpha_2, \alpha_3+\alpha_4, -\alpha_4, \alpha_5\biggr),
 \nonumber\\
 s_5 (*) =
 \biggl(x, y, z, w -\frac{\alpha_5}{z}, t ; \alpha_0, \alpha_1, \alpha_2, \alpha_3+\alpha_5, \alpha_4, -\alpha_5\biggr),\nonumber\\
 \pi_1(*) =
 (1-x, -y-t, 1-z, -w, t ; \alpha_1, \alpha_0, \alpha_2, \alpha_3, \alpha_5, \alpha_4),\nonumber\\
 \pi_2(*) =
 \biggl(\frac{y+w+t}{t}, -t (z-1), \frac{y+t}{t}, -t (x-z), -t ; \alpha_5, \alpha_4, \alpha_3, \alpha_2, \alpha_1, \alpha_0\biggr),
 \nonumber\\
 \pi_3(*) =
 (1-x, -y, 1-z, -w, -t ; \alpha_0, \alpha_1, \alpha_2, \alpha_3, \alpha_5, \alpha_4),\nonumber\\
 \pi_4(*) =
 (x, y+t, z, w, -t ; \alpha_1, \alpha_0, \alpha_2, \alpha_3, \alpha_4, \alpha_5) .\label{g}
 \end{gather}

 In fact, the actions $s_j$, $0 \leq j \leq 5$ define a representation of the affine Weyl group \smash{$W\bigl(D^{(1)}_5\bigr)$}.

 \begin{Remark}\label{tr}
 When $y=-t$ (resp.\ $y=0$) is a particular solution of the Sasano system~\eqref{s}, then the parameter
 $\alpha_0$ (resp.\ $\alpha_1$) must be equal to zero. So, when $y=-t$ (resp.\ $y=0$), we consider the
 transformation $s_0$ (resp.\ $s_1$) as an identity transformation.
 Note that $\alpha_1=0$ does not imply~${y=0}$ of necessity.
 For example, if $\alpha_1=0$, the function $y=-t$ is a particular solution of the system~\eqref{s}, provided that
 $\alpha_0=0$.
 Next, when $z=0$ (resp.\ $z=1$) is a particular solution of~\eqref{s}, we find that $\alpha_5=0$ (resp.\ $\alpha_4=0$).
 This time we consider the transformation~$s_5$ (resp.~$s_4$) as an identity transformation.
 Note that $\alpha_5=0$ does not imply $z=0$ of necessity. For example, when $\alpha_5=0$, the function
 $z=t$ is a particular solution of the system~\eqref{s}, provided that $\alpha_4=-1$ and $y=-\frac{t}{2}$,~$w=0$ solve
 the same system. When $w=0$,~$\alpha_3=0$ we consider the transformation $s_3$ as an identity transformation.
 Finally, when $x=z$,~$\alpha_2=0$,
 the transformation~$s_2$~is considered as an identity transformation.
 \end{Remark}

 Denote by $V:=(\alpha_0, \dots, \alpha_5)=(-\alpha, 0, 0, 0, 1, \alpha)$ the vector of parameters corresponding to the
 particular solution ${\rm Sol}:=(x, y, z, w)=\bigl(\frac{\alpha}{t}, 0, \frac{\alpha}{t}, 0\bigr)$. In order to
 describe the orbit of the vector~$V$ under the action of the group \smash{$\widetilde{W}\bigl(D^{(1)}_5\bigr)$}, we define
 using the ideas of Sasano the~following translation operators:
 \begin{alignat*}{4}
 &T_1=\pi_1 s_5 s_3 s_2 s_1 s_0 s_2 s_3 s_5,\qquad&&
 T_2=\pi_2 T_1 \pi_2, \qquad&& T_3=s_1 s_4 T_1 s_4 s_1, &\\
 &T_4=s_2 s_3 T_3 s_3 s_2,
 \qquad&& T_5=s_1 T_4 s_1,\qquad&&
 T_6=s_3 T_3 s_3 .&
 \end{alignat*}
 These operators act on the parameters as follows:
 \begin{gather*}
 T_1(\alpha_0, \alpha_1, \dots, \alpha_5) =
 (\alpha_0, \alpha_1, \dots, \alpha_5) + (0, 0, 0, 0, 1, -1),\\
 T_2(\alpha_0, \alpha_1, \dots, \alpha_5) =
 (\alpha_0, \alpha_1, \dots, \alpha_5) + (-1, 1, 0, 0, 0, 0),\\
 T_3(\alpha_0, \alpha_1, \dots, \alpha_5) =
 (\alpha_0, \alpha_1, \dots, \alpha_5) + (0, 0, 0, 1, -1, -1),\\
 T_4(\alpha_0, \alpha_1, \dots, \alpha_5) =
 (\alpha_0, \alpha_1, \dots, \alpha_5) + (1, 1, -1, 0, 0, 0),\\
 T_5(\alpha_0, \alpha_1, \dots, \alpha_5) =
 (\alpha_0, \alpha_1, \dots, \alpha_5) + (1, -1, 0, 0, 0, 0),\\
 T_6(\alpha_0, \alpha_1, \dots, \alpha_5) =
 (\alpha_0, \alpha_1, \dots, \alpha_5) + (0, 0, 1, -1, 0, 0) .
 \end{gather*}
 Note that the particular solution ${\rm Sol}$
 is transformed under the transformations $T_j$ as follows:
 \begin{gather*}
 T_1 ({\rm Sol}) =
 \biggl(\frac{\alpha}{t} - \frac{1}{t-\alpha+1}, 0, \frac{\alpha}{t} - \frac{1}{t-\alpha+1}, 0\biggr),\\
 T_2 ({\rm Sol}) =
 \biggl(1- \frac{1}{\alpha-t+1}, -1 + \frac{\alpha}{\alpha-t}, 1- \frac{1}{\alpha-t+1}, 0\biggr),\\
 T_3({\rm Sol}) =
 (1, 0, 1, -t+\alpha-1),\\
 T_4({\rm Sol}) =
 \biggl(1, -t + \alpha-1, 1 - \frac{1}{t-\alpha+1}, 0\biggr),\\
 T_5({\rm Sol}) =
 \biggl(1 - \frac{1}{t-\alpha+1}, -t+\alpha-1, 1- \frac{1}{t-\alpha+1}, 0\biggr),\\
 T_6({\rm Sol}) =
 \biggl(1 , 0, 1- \frac{1}{t-\alpha+1}, -t + \alpha-1\biggr) .
 \end{gather*}
 Note also that from Remark \ref{tr} it follows that all of the B\"{a}cklund transformations~\eqref{g} make sense for the particular solution
 ${\rm Sol}$.

 Denote
 $M_1=1-\alpha_0-\alpha_1$, $M_2=\alpha_4+\alpha_5$.
 The next two lemmas describe the orbit of the~vector~$V$
 under the group \smash{$\widetilde{W}\bigl(D^{(1)}_5\bigr)$} with generators~\eqref{g}.

 \begin{Lemma}\label{orbit}
 Assume that $\alpha\notin\mathbb{Z}$.
 Let $(x, y, z, t)=\bigl(\frac{\alpha}{t}, 0, \frac{\alpha}{t}, 0\bigr)$ be a particular rational solution
 of the Sasano system~\eqref{s} with vector of parameters $V$.
 Then applying B\"{a}cklund transformations~\eqref{g} to this solution, we obtain rational solutions of
 \eqref{s} with parameters $\alpha_j$, $0 \leq j \leq 5$, which are either of the kind
 $\pm \alpha + n_j$, $n_j\in\mathbb{Z}$ or of the kind $l_j$, $l_j\in\mathbb{Z}$ in such a way that
 $M_1$ and $M_2$ are together of the kind
 $\pm \alpha + m_j$, $m_j\in\mathbb{Z}$, $j=1,2$.
 \end{Lemma}

 \begin{proof}
 This lemma is proved by an induction on the numbers of the applied transformations~\eqref{g}
 on the vector $V$ and with the aid of above translation operators
 $T_j$, $1 \leq j \leq 6$.
\end{proof}

 \begin{Remark}
 	The conditions imposed on the parameters $\alpha_j$ by Lemma \ref{orbit} ensure that when ${\alpha\notin\mathbb{Z}}$,
 	at least two of the new-obtained parameters $\alpha_j$ are integer numbers and at least two of them are not integer numbers.
 \end{Remark}

 With the next lemma, we specify the orbit
 of the vector $V$ when $\alpha\in\mathbb{Z}$.

 \begin{Lemma}
 Assume that $\alpha\in\mathbb{Z}$.
 Let $(x, y, z, t)=\bigl(\frac{\alpha}{t}, 0, \frac{\alpha}{t}, 0\bigr)$ be a particular rational solution
 of the Sasano system~\eqref{s} with vector of parameters~$V$.
 Then applying B\"{a}cklund transformations~\eqref{g} to this solution we obtain rational solutions of
 \eqref{s}, for which all of the parameters~$\alpha_j$, $0 \leq j \leq 5$, are integer numbers
 in such a way that~$M_1$ and $M_2$ are together either even or odd integer.
 \end{Lemma}

 \begin{proof}
 This lemma is proved inductively.
 \end{proof}

 Following \cite{SY}, we define the symplectic transformations $r_i$, $0 \leq i \leq 5$, which correspond to the symmetries
 $s_i$, $0 \leq i \leq 5$, from~\eqref{g}
 \begin{gather*}
 r_0(x, y, z, w) =
 \biggl(\frac{1}{x}, -t - x (x (y+t) + \alpha_0), z, w\biggr),\\
 r_1(x, y, z, w) =
 \biggl(\frac{1}{x}, -x (x y + \alpha_1), z, w\biggr),\\
 r_2(x, y, z, w) =
 \biggl(z- y (y (x-z) -\alpha_2), \frac{1}{y}, z, w+y - \frac{1}{y}\biggr),\\
 r_3(x, y, z, w) =
 \biggl(x, y, \frac{1}{z}, -z (z w + \alpha_3)\biggr),\\
 r_4(x, y, z, w) =
 \biggl(x, y, 1-w (w (z-1) - \alpha_4), \frac{1}{w}\biggr),\\
 r_5(x, y, z, w) =
 \biggl(x, y, -w (w z - \alpha_5), \frac{1}{w}\biggr) .
 \end{gather*}
 The transformations $\pi_j$, $1 \leq j \leq 5$, from~\eqref{g} are also canonical symplectic transformations~since
 \[
 {\rm d}x \wedge {\rm d} y + {\rm d} z \wedge {\rm d} w={\rm d} \pi_j(x) \wedge {\rm d} \pi_j(y) + {\rm d} \pi_j (z) \wedge {\rm d} \pi_j(w) .
 \]
 Denote, in short, by $r_i$, $0 \leq i \leq 5$, and $\pi_j$, $1 \leq j \leq 4$, the image of the
 canonical coordinates~$x$,~$y$,~$z$,~$w$ under the action of $r_i$ and $\pi_j$, respectively.
 Then the following important result is an~extension of \cite[Theorem 4.1]{SY}.

 \begin{Theorem}\label{bct}
 There exists an unique polynomial Hamiltonian system of degree $4$, which is holomorphic in each coordinates
 $r_i$, $0 \leq i \leq 5$, and $\pi_j$, $1 \leq j \leq 5$. This system is invariant under the extended Weyl group
 \smash{$\widetilde{W}\bigl(D^{(1)}_5\bigr)$}
 and coincides with the system
 \eqref{s}.
 \end{Theorem}

 Theorem \ref{bct} says that the transformations $s_i$, $0 \leq i \leq 5$, and $\pi_j$, $1 \leq j \leq 4$, from~\eqref{g} are
 canonical transformations, which are rational on the functions $x$, $y$, $z$, $w$. Using this fact, we establish the main results of this
 paper.

 \begin{Theorem}\label{1}
 Let $\alpha$ be an arbitrary complex parameter, which is not an integer. Assume that the parameters $\alpha_j$ are either of the kind
 $\pm \alpha + n_j$ or of the kind $l_j, n_j, l_j\in\mathbb{Z}$ in such a way that~$M_1$ and
 $M_2$ are together of the kind $\pm \alpha + m_i$, $m_i\in\mathbb{Z}$, $i=1, 2$. Then the Sasano system~\eqref{s} is not integrable in the Liouville--Arnold sense by rational first integrals.
 \end{Theorem}

 \begin{Theorem}\label{2}
 Assume that all of the parameters $\alpha_j$, $0 \leq j \leq 5$, are integer in such a way that~$M_1$ and $M_2$ are together either even or odd integer. Then the Sasano system
 \eqref{s} is not integrable in the Liouville--Arnold sense by rational first integrals.
 \end{Theorem}

 \subsection*{Acknowledgments}
 The author is indebted to the referees for critical remarks and advices towards the improvement of the paper.
The author was partially supported by Grant KP-06-N 62/5 of the Bulgarian Fund
 ``Scientific research''.

\pdfbookmark[1]{References}{ref}
\LastPageEnding

\end{document}